\documentclass[preprint]{elsarticle}

\usepackage{amsmath}
\usepackage{amsthm}
\usepackage{cleveref}
\usepackage{enumitem}
\usepackage{color}
\usepackage{tikz}
\usepackage{amsfonts}
\usepackage{todonotes}
\usepackage{subfigure}
\usepackage{xspace}

\usepackage{setspace}

\setlist{nosep}

\newcommand{\imped}{impedensable\xspace}
\newcommand{\Imped}{Impedensable\xspace}
\newcommand{\mds}{MDS\xspace}
\newcommand{\ds}{DS\xspace}
\newcommand{\vc}{VC\xspace}
\newcommand{\is}{IS\xspace}
\newcommand{\mvc}{MVC\xspace}
\newcommand{\mis}{MIS\xspace}
\newcommand{\gc}{GC\xspace}
\newcommand{\tds}{2DS\xspace}
\newcommand{\sdmds}{SDMDS\xspace}

\newtheorem{algorithm}{Algorithm}
\newtheorem{definition}{Definition}
\newtheorem{lemma}{Lemma}
\newtheorem{theorem}{Theorem}
\newtheorem{corollary}{Corollary}
\newtheorem{example}{Example}
\newtheorem{observation}{Observation}

\newcommand{\arya}[1]{{\color{black} #1}}
\newcommand{\revisedtext}[1]{{\color{black} #1}}

\usepackage{amssymb}

\begin{document}

\title{Eventually Lattice-Linear Algorithms\footnote{\textbf{To appear in the Journal of Parallel and Distributed Computing.}}\footnote{\textbf{The experiments presented in this paper were supported through computational resources and services provided by the Institute for Cyber-Enabled Research, Michigan State University.}}\footnote{\textbf{A preliminary version of this paper was published in Proceedings of the 23rd International Symposium on Stabilization, Safety, and Security of Distributed Systems (SSS 2021) \cite{Gupta2021}.}}\footnote{Email addresses: \texttt{\{atgupta,sandeep\}@msu.edu}}}

\author{Arya Tanmay Gupta}
\author{Sandeep S Kulkarni}

\affiliation{Computer Science and Engineering, Michigan State University}

\begin{abstract}
    Lattice-linear systems allow nodes to execute asynchronously. We introduce eventually lattice-linear algorithms, where lattices are induced only among the states in a subset of the state space. The algorithm guarantees that the system transitions to a state in one of the lattices. Then, the algorithm behaves lattice linearly while traversing to an optimal state through that lattice.
    
    We present a lattice-linear self-stabilizing algorithm for service demand based minimal dominating set (SDMDS) problem. Using this as an example, we elaborate the working of, and define, eventually lattice-linear algorithms. Then, we present eventually lattice-linear self-stabilizing algorithms for minimal vertex cover (\mvc), maximal independent set (\mis), graph colouring (\gc) and 2-dominating set problems (\tds).
    
    Algorithms for SDMDS, \mvc and \mis converge in 1 round plus $n$ moves (within $2n$ moves), \gc in $n+4m$ moves, and \tds in 1 round plus $2n$ moves (within $3n$ moves).
    These results are an improvement over the existing literature.  
    We also present experimental results to show performance gain demonstrating the benefit of lattice-linearity. 
\end{abstract}


\maketitle

\section{Introduction}\label{section:introduction}

A parallel/distributed system consists of several processes that collectively solve a problem while coordinating with each other. 
In such a system, to successfully solve the given problem, processes need to coordinate with each other. This could be achieved via shared memory (where the data is stored in a centralized store and each process can access the part that it needs) or message passing (where processes send their updates to each other via messages). As the level of parallelization (number of threads/processes) increases, in these systems, the need for synchronization increases as well. 

Synchronization primitives such as local mutual exclusion and dining philosophers are deployed to achieve the required synchronization among processes. Without proper synchronization, the execution of a parallel/distributed algorithm may be incorrect. 
As an example, in an algorithm for graph colouring, in each step, a node reads the colour of its neighbours, and if necessary, updates its own colour. Execution of this algorithm in a parallel or distributed system requires synchronization to ensure correct behaviour. For example, if two neighbouring nodes, say $i$ and $j$ change their colour simultaneously, the resulting action may be incorrect.
Specifically, consider the case where both $i$ and $j$ have colour $0$. Both nodes read the states of each other. Consider that $i$ executes first, node $i$ executes and changes its colour to $1$. This is as expected. However, when $j$ executes, it has old/inconsistent information about the colour of $i$. Hence, $j$ also will change its colour to $1$. Effectively, without proper synchronization, $i$ and $j$ can end up choosing the same colour and repeat such execution forever. 
The synchronization primitives discussed above eliminate such behaviour. However, they introduce an overhead, which can be very costly in terms of computational resources and time.

In general, executing a parallel/distributed program without synchronization can lead to erroneous behaviour. However, if correctness can be guaranteed even in asynchrony, then we can benefit from concurrent execution without bearing the cost of synchronization. In turn, this will be highly useful as the level of parallelization increases. 

Garg \cite{Garg2020} introduced the modelling of problems using a predicate under which the global states induce a partial order forming a lattice. 
The problems that can be represented by a predicate under which the global states form a lattice are called \textit{lattice-linear problems}, and such a predicate is called \textit{lattice-linear predicate}. \arya{In such problems, given an arbitrary global state, such nodes can be pointed out that must change their state to satisfy the requirements of the problem -- to reach an optimal global state.} Modelling problems in such a way allows them to be solved by parallel processing systems under asynchrony (where a node may read old information about others) while preserving correctness.

However, \cite{Garg2020} cannot be applied to problems where a lattice-linear predicate does not exist. 
\arya{For example, in a general case of graph colouring, a specific node cannot be pointed out that must change its state to reach an optimal state: for an arbitrary node, it is possible to reach an optimal global state without changing its local state.}
Moreover, the problems studied in \cite{Garg2020} require that the system is initialized in a specific state. Hence, it is not applicable for systems that require self-stabilization.

In this paper, we study problems that require self-stabilization and cannot be represented by a lattice-linear predicate. 
We develop \textit{eventually lattice-linear self-stabilizing algorithms}  where one or more lattices are induced in a subset of the state space algorithmically. 
\revisedtext{In other words, there are subsets $S_1, S_2\cdots$ of the state space such that a lattice is induced in each subset.}
We present eventually lattice-linear self-stabilizing algorithms for service demand based minimal dominating set (\sdmds), minimal vertex cover (\mvc), maximal independent set (\mis), graph colouring (\gc) and 2-dominating set (\tds) problems.
These algorithms first (1) guarantee that from any arbitrary state, the system reaches a state in one of the lattices, and then (2) these algorithms behave lattice-linearly, and make the system traverse that induced lattice and reach an optimal state.

We proceed as follows. We begin with the SDMDS problem which is a generalization of the minimal dominating set. We devise a self-stabilizing algorithm for SDMDS. We scrutinize this algorithm and decompose it into two parts, the second of which satisfies the lattice-linearity property of \cite{Garg2020} if it begins in a \textit{feasible} state.  Furthermore, the first part of the algorithm ensures that the algorithm reaches a \textit{feasible} state. We show that the resulting algorithm is self-stabilizing, and the algorithm 
has 
\textit{limited-interference} property (to be discussed in \Cref{subsection:ds-eventual}) due to which it is tolerant to the nodes reading old values of other nodes. 
We also demonstrate that this approach is generic. It applies to various other problems including \mvc, \mis, \gc and \tds. 

We also present some experimental results that show the efficacy of eventually lattice-linear algorithms in real-time shared memory systems. \arya{Specifically, we compare our algorithm for \mis (\Cref{algorithm:rules-mis}) with algorithms presented in \cite{Hedetniemi2003} and \cite{Turau2007}. The experiments are conducted in \texttt{cuda} environment, which is a shared memory model.}

\subsection{Contributions of the paper}

\begin{itemize}
    \item We present a self-stabilizing algorithm for the SDMDS problem. 
    \item We extend the theory of \textit{lattice-linear predicate detection} from \cite{Garg2020} to introduce the class of \textit{eventually lattice-linear self-stabilizing algorithms}.
    Such algorithms allow the system to run in asynchrony and ensure its convergence even when the nodes read old values. \arya{The algorithms presented in \cite{Garg2020} also allow asynchrony, however, they require that (1) the problems have only one optimal state, and (2)  the algorithm needs to start in specific initial states.}
    \item We show that the general design used to develop eventually lattice-linear algorithm for \sdmds can be extended to other problems such as
    \mvc, \mis, \gc and \tds problems.
    \item The algorithms for \sdmds, \mvc and \mis converge in 1 round plus $n$ moves, the algorithm for \gc converges in $n+4m$ moves, and the algorithm for \tds converges in 1 round plus $2n$ moves. Also, these algorithms do not require a synchronous environment to execute. Thus, these results are an improvement over the algorithms present in the literature.
    \item We use the algorithm for \mis, as an illustration, to show that this algorithm outperforms existing algorithms. 
    \arya{This is because these algorithms guarantee convergence in asynchrony. This, in turn, is possible because of the theoretical guarantees due to the lattice-linear nature of these algorithms. The maximum benefit of the asynchrony, allowed by such algorithms, will be observed in higher levels of concurrency, e.g., if only one thread performs execution for a node of the input graph.}
    
\end{itemize}

\noindent The major focus and benefit of this paper is inducing lattices algorithmically in non-lattice-linear problems to allow asynchronous executions. Some applications of the specific problems studied in this paper are listed as follows.
Dominating set is applied in communication and wireless networks where it is used to compute the virtual backbone of a network.
Vertex cover is applicable in (1) computational biology, where it is used to eliminate repetitive DNA sequences -- providing a set covering all desired sequences, and (2) economics, where it is used in camera instalments -- it provides a set of locations covering all hallways of a building.
Independent set is applied in computational biology, where it is used in discovering stable genetic components for designing engineered genetic systems.
Graph colouring is applicable in (1) chemistry, where it is used to design storage of chemicals -- a pair of reacting chemicals are not stored together,
and (2) communication networks, where it is used in wireless networks to compute radio frequency assignment.

\subsection{Organization of the paper}

This paper is organized as follows. In \Cref{section:preliminaries}, we describe some notations and definitions that we use in the paper.
In \Cref{section:sdmds-algorithm}, we describe the algorithm for the service demand based minimal dominating set problem. In \Cref{section:sdmds-lattice-linear}, we analyze the characteristics of that algorithm and show that it is eventually lattice-linear. We use the structure of eventually lattice-linear self-stabilizing algorithms to develop algorithms for minimal vertex cover, maximal independent set, graph colouring and 2-dominating set problems, respectively, in Sections \ref{section:mvc}, \ref{section:mis}, \ref{section:gc} and \ref{section:2ds}. We discuss the related work in \Cref{section:literature}. Then, in \Cref{section:experiments}, we compare the convergence speed of the algorithm presented in \Cref{section:mis} with other algorithms (for the maximal independent set problem) in the literature \arya{(specifically, \cite{Hedetniemi2003} and \cite{Turau2007})}.
Finally, we conclude in \Cref{section:conclusion}.

\section{Preliminaries}\label{section:preliminaries}

Throughout the paper, we denote $G$ to be an arbitrary 
undirected graph on which we apply our algorithms. $V(G)$ is the vertex-set and $E(G)$ is the edge-set of $G$. For any node $i$, $Adj_i$ is the set of nodes connected to $i$ by an edge in $G$.
$deg(i)$ equals $|Adj_i|$. For a natural number $x$, $[1:x]$ is the sequence of natural numbers from 1 to $x$.

Each node $i$ is associated with a set of variables. The algorithms are written in terms of rules, where each \textit{rule} for process $i$ is of the form $g\longrightarrow a_c$ where \arya{(1)} the \textit{guard} $g$ is a proposition over variables of 
some nodes which may include the variables of $i$ itself along with the variables of other nodes\arya{, and (2) the \textit{action} $a_c$ is a set of instructions that updates the variables of $i$ if $g$ is true}.
If any of the guards hold true for some node, we say that the node is \textit{enabled}. 

A \textit{move} is an event in which an enabled node updates its variables by executing an action $a_c$ corresponding to a guard $g$ that is true.
A \textit{round} is a sequence of events in which every node evaluates its guards at least once, and makes a move accordingly.

In some of the synchronization models, a selected node acts as a scheduler for the rest of the processes. A \textit{scheduler/daemon} is a node whose function is to choose one, some, or all nodes in a time step, throughout the execution, so that the selected nodes can evaluate their guards and take the corresponding action. A \textit{central scheduler} chooses only one node per time step. A \textit{distributed scheduler} one or more nodes, possibly arbitrarily, per time step. A \textit{synchronous scheduler} chooses all the nodes in each time step.

A global state $s\in S$ is represented as a vector such that $s[i]$ denotes the variables of node $i$, $s[i]$
itself is a vector of the variables of node $i$.

An algorithm $A$ is \textit{self-stabilizing} with respect to the subset $S_o$ of $S$ iff (1) \textit{convergence}: starting from an  arbitrary state, any sequence of computations of $A$ reaches a state in $S_o$, and (2) \textit{closure}: any computation of $A$ starting from $S_o$ always stays in $S_o$. 
We assume $S_o$ to be the set of \textit{optimal} states: the system is deemed converged once it reaches a state in $S_o$. $A$ is a \textit{silent} self-stabilizing algorithm if no node is enaabled once a state in $S_o$ is reached.

\subsection{Execution without Synchronization}

Typically, we view the \textit{computation} of an algorithm as a sequence of global states $\langle s_0, s_1, \cdots\rangle$, where $s_{t+1}, t\geq 0,$ is obtained by executing some action by one or more nodes (as decided by the scheduler) in $s_t$.  
For the sake of discussion, assume that only node $i$ executes in state $s_t$. 
The computation prefix til $s_{t}$ is $\langle s_0, s_1, \cdots, s_t\rangle$. The state that the system traverses to after $s_t$ is $s_{t+1}$.
Under proper synchronization, $i$ would evaluate its guards on the \textit{current} local states of its neighbours in $s_t$, and the resultant state $s_{t+1}$ can be computed accordingly.

To understand the execution in asynchrony, let $x(s)$ be the value of some variable $x$ 
in state $s$. 
If $i$ executes in asynchrony, then it views the global state that it is in to be $s'$, 
where $x(s')\in\{ x(s_0), x(s_1), \cdots, x(s_t) \}.$
In this case, $s_{t+1}$ is evaluated as follows.
If all guards in $i$ evaluate to false, then the system will continue to remain in state $s_t$, i.e., $s_{t+1} = s_{t}$.
If a guard $g$ evaluates to true then $i$ will execute its corresponding action $a_c$.
Here, we have the following observations:
(1) $s_{t+1}[i]$ is the state that $i$ obtains after executing an action in $s'$, and (2) $\forall j\neq i$, $s_{t+1}[j] = s_t[j]$.

\arya{
The model described in the above paragraph is \textit{arbitrary asynchrony}, in which a node can read old values of other nodes arbitrarily, requiring that if some information is sent from a node, it eventually reaches the target node.
In this paper, however, we are interested in \textit{asynchrony with monotonous read} (AMR) model. 
AMR model is arbitrary asynchrony with an additional restriction: when node $i$ reads the state of node $j$, the reads are monotonic, i.e., if $i$ reads a newer value of the state of $j$ then it cannot read an older value of $j$ at a later time. E.g., if the state of $j$ changes from $0$ to $1$ to $2$ and node $i$ reads the state of $j$ to be $1$ then its subsequent read will either return $1$ or $2$, it cannot return $0$. 
}

\subsection{Embedding a $\prec$-lattice in global states}

Let $s$ denote a global state, and let $s[i]$ denote the state of node $i$ in $s$. First, we define a total order $\prec_l$; all local states of a node $i$ are totally ordered under $\prec_l$. 
Using $\prec_l$, we define a partial order $\prec_g$ among global states as follows. 

We say that $s \prec_g s^\prime$ iff $(\forall i: s[i]=s'[i]\lor s[i]\prec_l s'[i]) \land (\exists i:s[i]\prec_ls'[i])$.
Also, $s=s'$ iff $\forall i~s[i] = s'[i]$. 
For brevity, we use $\prec$ to denote $\prec_l$ and $\prec_g$: $\prec$ corresponds to $\prec_l$ while comparing local states, and $\prec$ corresponds to $\prec_g$ while comparing global states. 
We also use the symbol `$\succ$' which is such that $s\succ s'$ iff $s' \prec s$.
Similarly, we use symbols `$\preceq$' and `$\succeq$'; e.g., $s\preceq s'$ iff  $s=s' \vee s \prec s'$.
We call the lattice, formed from such partial order, a \textit{$\prec$-lattice}.

\begin{definition}\label{definition:<-lattice}
    \textbf{{\boldmath$\prec$}-\textit{lattice}}. 
    Given a total relation $\prec_l$ that orders the states visited by a node $i$ (for each $i$) the $\prec$-lattice corresponding to $\prec_l$ is defined by the following partial order:
    $s \prec_g s'$ iff $(\forall i \ \ s[i] \preceq_l s'[i]) \wedge (\exists i \ \ s[i] \prec_l s'[i])$.
\end{definition}

In the $\prec$-lattice discussed above, we can define the meet and join of two states in the standard way: the meet (respectively, join), of two states $s_1$ and $s_2$ is a state $s_3$ where $\forall i, s_3[i]$ is equal to $min(s_1[i], s_2[i])$ (respectively, $max(s_1[i], s_2[i])$), where $\min(x, y) = \min(y, x)=x$ iff $(x\prec_l y \lor x=y)$, and 
 $\max(x, y) = \max(y, x)=y$ iff $(y\succ_l x \lor y=x)$.
 For $s_1$ and $s_2$, their meet (respectively, join) has paths to (respectively, is reachable from) both $s_1$ and $s_2$.

A $\prec$-lattice, embedded in the state space, is useful for permitting the algorithm to execute asynchronously.
Under proper constraints on the structure of $\prec$-lattice, convergence can be ensured. 

\arya{
\subsection{Lattice-Linear Problems and Algorithms}

Next, we discuss \textit{lattice-linear problems}, i.e., the problems where the description of the problem statement creates the lattice structure automatically. Such problems can be represented by a predicate under which the states in $S$ form a lattice. These problems include stable (man-optimal) marriage problem, market clearing price and others. These problems have been discussed in \cite{Garg2020, Garg2021, Garg2022}. 

A \textit{lattice-linear problem} $P$ can be represented by a predicate $\mathcal{P}$ such that if any node $i$ is violating $\mathcal{P}$ in a state $s$, then it must change its state. Otherwise, the system will not satisfy $\mathcal{P}$.
Let $\mathcal{P}(s)$ be true iff state $s$ satisfies $\mathcal{P}$. A node violating $\mathcal{P}$ in $s$ is called an \textit{\imped} node (an \textit{impediment} to progress if does not execute, \textit{indispensable} to execute for progress). Formally,

\begin{definition}\label{definition:impedensable-node}\cite{Garg2020} \textbf{\textit{\Imped node.}} $\textsc{\Imped}(i,s,\mathcal{P})\equiv \lnot \mathcal{P}(s)$ $\land$ $(\forall s'\succ s:s'[i]=s[i]\Rightarrow\lnot \mathcal{P}(s'))$. \end{definition}

\noindent\textbf{\textit{Remark}}: We use the term `impedensable' as a combination of the English words impediment and indispensable.
The term `\imped' is similar to the notion of a node being \textit{forbidden} introduced in \cite{Garg2020}. This word itself comes from predicate detection background \cite{Chase1995}. 
We changed the notation to avoid the misinterpretation of the English meaning of the word `forbidden'.

If a node $i$ is \imped in state $s$, then in any state $s'$ such that $s'\succ s$, if the state of $i$ remains the same, then the algorithm will not converge.
Thus, predicate $\mathcal{P}$ induces a total order among the local states visited by a node, for all nodes. Consequently, the discrete structure that gets induced among the global states is a $\prec$-lattice, as described in \Cref{definition:<-lattice}. 
We say that $\mathcal{P}$, satisfying \Cref{definition:impedensable-node}, is \textit{lattice-linear} with respect to that $\prec$-lattice.

There can be multiple arbitrary lattices that can be induced among the global states. A system cannot guarantee convergence while traversing an arbitrary lattice. To guarantee convergence, we design the predicate $\mathcal{P}$ such that it fulfils some properties, and guarantees reachability to an optimal state. $\mathcal{P}$ is used by the nodes to determine if they are \imped, using \Cref{definition:impedensable-node}.
Thus, any suboptimal global state has at least one \imped node.

\begin{definition}\cite{Garg2020}\textbf{\textit{Lattice-linear predicate.}}
    $\mathcal{P}$ is a lattice-linear predicate with respect to a $\prec$-lattice induced among the global states iff $\forall s\in S: \lnot\mathcal{P}(s) \Rightarrow \exists i:\textsc{\Imped}(i,s,\mathcal{P})$.
\end{definition}

Now we complete the definition of lattice-linear problems. In a lattice-linear problem $P$, given any suboptimal global state $s$, we can identify all and the only nodes which cannot retain their local states. 
$\mathcal{P}$ is thus designed conserving this nature of the subject problem $P$.

\begin{definition}\label{definition:ll-problem}
    \textbf{Lattice-linear problems}.
    A problem $P$ is lattice-linear 
    iff there exists a predicate $\mathcal{P}$ and a $\prec$-lattice such that
    
    \begin{itemize}
        \item $P$ is deemed solved iff the system reaches a state where $\mathcal{P}$ is true,
        \item $\mathcal{P}$ is lattice-linear with respect to the $\prec$-lattice induced among the states in $S$, i.e., $\forall s: \neg \mathcal{P}(s) \Rightarrow \exists i:\textsc{\Imped}(i,s,\mathcal{P})$, and
        \item $\forall s:(\forall i:\textsc{\Imped}(i,s,\mathcal{P})\Rightarrow (\forall s':\mathcal{P}(s')\Rightarrow s'[i]\neq s[i]))$.
    \end{itemize}
\end{definition}

\noindent\textbf{\textit{Remark:}} A $\prec$-lattice, induced under $\mathcal{P}$, allows asynchrony because if a node, reading old values, reads the current state $s$ as $s'$, then $s'\prec s$. So $\lnot\mathcal{P}(s')\Rightarrow \lnot\mathcal{P}(s)$ because $\textsc{\Imped}(i,s',\mathcal{P})$ and $s'[i]=s[i]$.

\begin{definition}\label{definition:ssll-problem}
    \textbf{Self-stabilizing lattice-linear predicate}.
    Continuing from \Cref{definition:ll-problem},
    $\mathcal{P}$ is a self-stabilizing lattice-linear predicate if and only if the supremum of the lattice, that $\mathcal{P}$ induces, is an optimal state.
\end{definition}

\noindent Note that a self-stabilizing lattice-linear predicate $\mathcal{P}$ can also be true in states other than the supremum of the $\prec$-lattice. 

Certain problems are \textit{non-lattice-linear problems}. In such problems, there are instances in which the \imped nodes cannot be determined naturally, i.e., in those instances
$\exists s :\lnot\mathcal{P}(s) \wedge   (\forall i : \exists s' : \mathcal{P}(s')\land s[i]=s'[i]$).
For such problems, $\prec$-lattices can be induced algorithmically, through \textit{lattice-linear algorithms}.

\begin{definition}\label{definition:ll-algos}\textbf{Lattice-linear algorithms (LLA)}.
    Algorithm $A$ is an LLA for a problem $P$, iff there exists a predicate $\mathcal{P}$ and $A$ induces a $\prec$-lattice among the states of $S_1, ..., S_w \subseteq S (w\geq 1)$, such that
    \begin{itemize}
        \item State space $S$ of $P$ contains mutually disjoint lattices, i.e.
        \begin{itemize}
            \item $S_1, S_2, \cdots, S_w\subseteq S$ are pairwise disjoint.
            \item $S_1 \cup S_2 \cup \cdots \cup S_w$ contains all the reachable states (starting from a set of initial states, if specified; if an arbitrary state can be an initial state, then $S_1 \cup S_2 \cup \cdots \cup S_w=S$).
        \end{itemize}
        \item Lattice-linearity is satisfied in each subset under $\mathcal{P}$, i.e., 
        \begin{itemize}
            \item $P$ is deemed solved iff the system reaches a state where $\mathcal{P}$ is true
            \item $\forall k$, $1 \leq k \leq w$, 
            $\mathcal{P}$ is lattice-linear with respect to the partial order induced in $S_k$ by $A$, i.e., $\forall s\in S_k: \lnot\mathcal{P}(s) \Rightarrow \exists i \ \
            \textsc{\Imped}(i,s,\mathcal{P})$.
        \end{itemize}
    \end{itemize}
\end{definition}

\begin{definition}\textbf{Self-stabilizing LLA}.
    Continuing from \Cref{definition:ll-algos}, $A$ is self-stabilizing only if $S_1 \cup S_2 \cup \cdots \cup S_w=S$ and $\forall k:1\leq k\leq w$, the supremum of the lattice induced among the states in $S_k$ is optimal.
\end{definition}

In this paper, we study algorithms that induce one or more lattices in a subset of state space; they also guarantee that the system reaches a state in one of the induced lattices from an arbitrary state.
}

\section{Service Demand based Minimal Dominating Set}\label{section:sdmds-algorithm}

In this section, 
we introduce a generalization of the minimal dominating set (\mds) problem (\Cref{subsection:pd:ds}),
the service demand based minimal dominating set (\sdmds) problem, and describe an algorithm to solve it (\Cref{subsection:ds-general-algorithm}). 

\subsection{Problem description}\label{subsection:pd:ds}

The \sdmds problem, a generalization of \mds, \arya{is a simulation, on an arbitrary graph $G$, in which all nodes have some demands to be fulfilled and they offer some services. If a node $i$ is in the dominating set then it can not only serve all its own demands $D_i$, but also offer services from, its set of services $S_i$, to its neighbours. If $i$ is not in the dominating set, then it is considered dominated only if each of its demands in $D_i$ is being served by at least one of its neighbours that is in the dominating set.}

\begin{definition}\textbf{Service demand based minimal dominating set problem (SDMDS)}.
    In the \textit{service demand based minimal dominating set} problem, the input is a graph $G$ and a set of services $S_i$ and a set of demands $D_i$ for each node $i$ in $G$; the task is to compute a minimal set $\mathcal{D}$ such that for each node $i$,
    \begin{enumerate}
        \item either $i\in \mathcal{D}$, or
        \item for each demand $d$ in $D_i$, there exists at least one node $j$ in $Adj_i$ such that $d\in S_j$ and $j\in \mathcal{D}$.
    \end{enumerate}
\end{definition}

In the above generalization of the \mds problem, if all nodes have same set $X$ as their services and demands, i.e., $\forall i: S_i=X$ and $D_i=X$, then it is equivalent to \mds.

In the following subsection, we present a self-stabilizing algorithm for the minimal SDMDS problem.
Each node $i$ is associated with variable $i[st]$ with domain $\{IN, OUT\}$. $i[st]$ defines the state of $i$. We define $\mathcal{D}$ to be the set $\{i\in V(G): i[st]=IN\}$. 

\subsection{Algorithm for SDMDS problem}\label{subsection:ds-general-algorithm}

The list of constants, provided with the input, is in the following table.

\begin{center}
    \begin{tabular}{|l|l|}
        \hline
        Constant & What it stands for\\
        \hline
        $D_i$ & the set of demands of node $i$.\\
        $S_i$ & the set of services provided by node $i$.\\
        \hline
    \end{tabular}
\end{center}

The macros that we utilize are described in the following table.
Recall that $\mathcal{D}$ is the set of nodes which currently have the state as $IN$. A node $i$ is \textit{addable} if there is at least one demand of $i$ that is not being serviced by any neighbour of $i$ that is in $\mathcal{D}$.
A node $i$ is \textit{removable} if $\mathcal{D}\setminus\{i\}$ is also a dominating set given that $\mathcal{D}$ is a dominating set.
The \textit{dominators} of $i$ are the nodes that are (possibly) dominating node $i$: if some node $j$ is in \textsc{Dominators-Of}($i$), then there is at least one demand $d\in D_i$ such that $d\in S_j$.
$i$ is \textit{\imped} if $i$ is removable and there is no node $k$ that is removable and is of an ID higher than $i$, such that $k$ and $i$ are able to serve for some common node $j$.

\begin{center}
    \begin{tabular}{|l|}
        \hline
        $\mathcal{D}\equiv \{i\in V(G): i[st] =$ $IN\}$.\\
        \textsc{Addable-SDMDS}($i)\equiv i[st]=OUT\land$\\
        \quad\quad\quad\quad $(\exists d\in D_i, \forall j\in Adj_i: d\not\in S_j\lor j[st]=OUT)$.\\
        \textsc{Removable-SDMDS}$ (i)\equiv (\forall d \in D_i : (\exists j \in Adj_i: d \in S_j \land j[st]=$ $IN))\land$\\
        \quad\quad\quad\quad $(\forall j \in Adj_i,\forall~d \in D_j:d\in S_i\Rightarrow$\\
        \quad\quad\quad\quad $(\exists k \in Adj_j, k\neq i:(d \in S_k \land k[st] =$ $IN)))$.\\
        \textsc{Dominators-Of}($i)\equiv$\\
        \quad\quad\quad\quad $\{j\in Adj_i, j[st]=IN:\exists d\in D_i:d\in S_j\}\cup\{i\}$\quad if $i[st]=IN$\\
        \quad\quad\quad\quad $\{j\in Adj_i, j[st]=IN:\exists d\in D_i:d\in S_j\}$\quad \quad \quad \quad otherwise.\\
        \textsc{Impedensable-SDMDS}$(i)\equiv i[st]=IN\land$ \textsc{Removable-DS}$(i)\land$\\
        \quad\quad\quad\quad $(\forall j \in Adj_i,\forall~d \in D_j:d\in S_i\Rightarrow$\\
        \quad\quad\quad\quad $((\forall k \in$ \textsc{Dominators-Of}$(j)$, $k\neq i:(d \in S_k \land k[st] =$ $IN))\Rightarrow$\\
        \quad\quad\quad\quad $(k[id]<i[id]\lor\lnot$\textsc{Removable-DS}$(k))))$.\\
        \hline
    \end{tabular}
\end{center}

The general idea our algorithm is as follows. 
\begin{enumerate}
    \item A node enters the dominating set unconditionally if it is addable.
    This ensures that $G$ enters a state where the set of nodes in $\mathcal{D}$ form a (possibly non-minimal) dominating set. 
    If $\mathcal{D}$ is a dominating set, we say that the corresponding state is a \textit{feasible} state. 
    \item While entering the dominating set is not lattice-linear, the instruction governing the leaving of the dominating set is lattice-linear. 
    Node $i$ leaves the dominating set iff it is \imped. 
    Specifically, if $i$ serves for a demand $d$ in $D_j$ where $j \in Adj_i$ and the same demand is also served by another node $k$ ($k\in Adj_j$) then $i$ leaves only if (1) $k[id] < i[id]$ or (2) $k$ is not removable. 
    This ensures that if some demand $d$ of $D_j$ is satisfied by both $i$ and $k$ both of them cannot leave the dominating set simultaneously. 
    This ensures that $j$ will remain dominated. 
\end{enumerate} 
Thus, the rules for \Cref{algorithm:rules-ds} are as follows: 

\begin{algorithm}\label{algorithm:rules-ds}Rules for node $i$.
    \begin{center}
        \begin{tabular}{|l|}
            \hline
            \textsc{Addable-SDMDS}$(i)\longrightarrow i[st]=IN$.\\
            \textsc{Impedensable-SDMDS}$(i)\longrightarrow i[st]=OUT$.\\
            \hline
        \end{tabular}
    \end{center}
\end{algorithm}

We decompose  \Cref{algorithm:rules-ds} into two parts: (1) \Cref{algorithm:rules-ds}.1, that only consists of first guard and action of \Cref{algorithm:rules-ds} and (2) \Cref{algorithm:rules-ds}.2, that only consists of the second guard and action of \Cref{algorithm:rules-ds}. We use this decomposition in the following section of this paper to relate the algorithm to eventual lattice-linearity.

\section{Lattice-Linear Characteristics of \Cref{algorithm:rules-ds}}\label{section:sdmds-lattice-linear}

In this section, we analyze the characteristics of \Cref{algorithm:rules-ds} to demonstrate that it is eventually lattice-linear. We proceed as follows.
In \Cref{subsection:ds-propositions}, we state the propositions which define the feasible and optimal states of the SDMDS problem, along with some other definitions. In \Cref{subsection:guarantee-feasible}, we show that $G$ reaches a state where it manifests a (possibly non-minimal) dominating set.
In \Cref{subsection:ds-action2}, we show that after when $G$ reaches a feasible state, \Cref{algorithm:rules-ds} behaves like a lattice-linear algorithm. 
In \Cref{subsection:termination}, we show that when $\mathcal{D}$ is a minimal dominating set, no nodes are enabled. 
In \Cref{subsection:ds-eventual}, we argue that because there is a bound on interference between \Cref{algorithm:rules-ds}.1 and \Cref{algorithm:rules-ds}.2 even when the nodes read old values, \Cref{algorithm:rules-ds} is an eventually lattice-linear self-stabilizing (ELLSS) algorithm.
In \Cref{subsection:time-space-complexity-analysis}, we study the time and space complexity attributes of \Cref{algorithm:rules-ds}.

\subsection{Propositions stipulated by the SDMDS problem}\label{subsection:ds-propositions}

The SDMDS problem stipulates that the nodes whose state is $IN$ must collectively form a dominating set. We represent this proposition as $\mathcal{P}_{sdmds}^\prime$.
\begin{center}
    $\mathcal{P}_{sdmds}^\prime(\mathcal{D}) \equiv \forall i\in V(G):(i\in \mathcal{D}\lor (\forall d\in D_i,\exists j\in Adj_i: (d\in S_j\land j\in \mathcal{D})))$.
\end{center}
The SDMDS problem stipulates an additional condition that $\mathcal{D}$ should be a minimal dominating set. We represent this proposition as $\mathcal{P}_{sdmds}$.

\begin{center}
    $\mathcal{P}_{sdmds}(\mathcal{D})\equiv \mathcal{P}^\prime_{sdmds}(\mathcal{D})\land(\forall i\in \mathcal{D}, \lnot\mathcal{P}_{sdmds}^\prime(\mathcal{D}\setminus\{i\}))$.
\end{center}

If $\mathcal{P}_{sdmds}^\prime(\mathcal{D})$ is true, then $G$ is in a \textit{feasible} state. And, if $\mathcal{P}_{sdmds}(\mathcal{D})$ is true, then $G$ is in an \textit{optimal} state. 

Based on the above definitions, we define two scores with respect to the global state, \textsc{Rank} and \textsc{Badness}.
\textsc{Rank} determines the number of nodes needed to be added to $\mathcal{D}$ to change $\mathcal{D}$ to a dominating set. 
\textsc{Badness} determines the number of nodes that are needed to be removed from $\mathcal{D}$ to make it a minimal dominating set, given that $\mathcal{D}$ is a (possibly non-minimal) dominating set.

\begin{definition}
    $\textsc{Rank}(\mathcal{D})\equiv\min\{|\mathcal{D}^\prime|-|\mathcal{D}|:\mathcal{P}^\prime_d(\mathcal{D}^\prime)\land \mathcal{D}\subseteq \mathcal{D}^\prime\}$.
\end{definition}

\begin{definition}
    $\textsc{Badness}(\mathcal{D})\equiv\max\{|\mathcal{D}|-|\mathcal{D}^\prime|:\mathcal{P}^\prime_d(\mathcal{D}^\prime)\land \mathcal{D}^\prime\subseteq \mathcal{D}\}$.
\end{definition}

\subsection{Guarantee to Reach a Feasible State by \Cref{algorithm:rules-ds}}\label{subsection:guarantee-feasible}

We show that under \Cref{algorithm:rules-ds}, $G$ is guaranteed to reach a feasible state. 

\begin{lemma}\label{lemma:d-not-ds}
    Let $t.\mathcal{D}$ be the value of $\mathcal{D}$ at the beginning of round $t$. 
If $t.\mathcal{D}$ is not a dominating set then $(t+1).\mathcal{D}$ is a dominating set.
\end{lemma}

\begin{proof}
    Let $i$ be a node such that $i\in t.\mathcal{D}$ and $i\not\in(t+1).\mathcal{D}$, {i.e., $i$ leaves the dominating set in round $t$}. This means that $i$ remains dominated and all nodes in $Adj_i$ remain dominated, even when $i$ is removed. This implies that $i$ will not reduce the feasibility of $t.D$; it will not increase the value of \textsc{Rank}.
    
    Now let $\ell$ be a node such that $\ell\not \in t.\mathcal{D}$ which is addable when it evaluates its guards in round $t$. This implies that $\exists~d\in D_\ell$ such that $d$ is not present in $S_j$ for any $j\in Adj_\ell$ that is in the dominating set. According to the algorithm, the guard of the second action is true for $\ell$. This implies that $\ell[st]$ will be set to $IN$. 
    
    It can also be possible for the node $\ell$ that it is not addable when it evaluates its guards in round $t$. This may happen if some other nodes around $\ell$ already decided to move to $\mathcal{D}$, and as a result $\ell$ is now dominated. Hence $\ell\not\in(t+1).\mathcal{D}$ and we have that $\ell$ is dominated at round $t+1$.
    
    Therefore, we have that $(t+1).\mathcal{D}$ is a dominating set, which may or may not be minimal.
\end{proof}

From \Cref{lemma:d-not-ds}, we have that if at the beginning of some round, $G$ is in a state where $\textsc{Rank} >0$, then by the end of that round, \textsc{Rank} will be $0$.

\subsection{Lattice-Linearity of \Cref{algorithm:rules-ds}.2}\label{subsection:ds-action2}

In the following lemma, we show that \Cref{algorithm:rules-ds}.2 is lattice-linear.

\begin{lemma}\label{lemma:ds-addition}
    If $t.\mathcal{D}$ is a non-minimal dominating set then under \Cref{algorithm:rules-ds} (more specifically, \Cref{algorithm:rules-ds}.2),
    there exists at least one node such that $G$ cannot reach a minimal dominating set until that node is removed from the dominating set.
\end{lemma}

\begin{proof}
Since $\mathcal{D}$ is a dominating set, the first guard is false for all nodes in $G$.

Since $\mathcal{D}$ is not minimal, there exists at least one node that must be removed in order to make $\mathcal{D}$ minimal. Let $S^\prime$ be the set of nodes which are removable. 
Let $M$ be some node in $S^\prime$. If $M$ is not serving any node, then \textsc{Impedensable-SDMDS}($M$) is trivially true. Otherwise there exists at least one node $j$ which is served by $M$, that is, $\exists d\in D_j:d\in S_M$. We study two cases which are as follows:
(1) for some
node $j$ served by $M$, there does not exist another node $b \in S^\prime$ which serves $j$, and (2) for any node $b \in S^\prime$ such that $M$ and $b$ serve some common node $j$, $b[id]<M[id]$. 

In the first case, $M$ cannot be removed because \textsc{Impedensable-SDMDS}($M$) is false and, hence, $M$ cannot be in $S^\prime$, thereby leading to a contradiction.
In the second case, \textsc{Impedensable-SDMDS}($M$) is true and \textsc{Impedensable-SDMDS}($b$) is false since $b[id] < M[id]$. Thus, node $b$ cannot leave the dominating set until $M$ leaves. In both the cases, we have that $j$ stays dominated.

Since ID of every node is distinct, we have that there exists at least one node $M$ for which \textsc{Impedensable-SDMDS}($M$) is true. For example, \textsc{Impedensable-SDMDS} is true for the node with the highest ID in $S^\prime$; $G$ cannot reach a minimal dominating set until $M$ is removed from the dominating set.
\end{proof}

From \Cref{lemma:ds-addition}, it follows that \Cref{algorithm:rules-ds}.2 satisfies the condition of lattice-linearity as described in \Cref{section:preliminaries}. It follows that if we start from a state where $\mathcal{D}$ is a (possibly non-minimal) dominating set and execute \Cref{algorithm:rules-ds}.2 then it will reach  a state where $\mathcal{D}$ is a minimal dominating set even if nodes are executing with old information about others. Next, we have the following result which follows from \Cref{lemma:ds-addition}. 

\begin{lemma}\label{lemma:ds-removal}
    Let $t.\mathcal{D}$ be the value of $\mathcal{D}$ at the beginning of round $t$. 
    If $t.\mathcal{D}$ is a non-minimal dominating set then $|(t+1).\mathcal{D}| \leq |t.\mathcal{D}|-1$, and $(t+1).\mathcal{D}$ is a dominating set.
\end{lemma}

\begin{proof}
From \Cref{lemma:ds-addition}, {at least one node $M$ (including the maximum ID node in $S^\prime$} from the proof of \Cref{lemma:ds-addition}) would be removed  in round $t$. Furthermore, since $\mathcal{D}$ is a dominating set, \textsc{Addable}($i$) is false at every node $i$. Thus, no node is added to $\mathcal{D}$ in round $t$. Thus, the $|(t+1).\mathcal{D}|\leq |t.\mathcal{D}|-1$. 
    
    For any node $M$ that is removable, \textsc{Impedensable-SDMDS}($i$) is true only if any node $j$ which is (possibly) served by $M$ has other neighbours (of a lower ID) which serve the demands which $M$ is serving to it. This guarantees that $j$ stays dominated and hence $(t+1).\mathcal{D}$ is a dominating set.
\end{proof}

\subsection{Termination of \Cref{algorithm:rules-ds}}\label{subsection:termination}

The following lemma studies the action of \Cref{algorithm:rules-ds} when $\mathcal{D}$ is a minimal dominating set.

\begin{lemma}\label{lemma:ds-minimal}
Let $t.\mathcal{D}$ be the value of $\mathcal{D}$ at the beginning of round $t$. If $\mathcal{D}$ is a minimal dominating set, then $(t+1).\mathcal{D}=t.\mathcal{D}$.

\end{lemma}

\begin{proof}
    Since $\mathcal{D}$ is a dominating set, \textsc{Addable}($i$) is false for every node in $V(G)$, i.e., the first action is disabled for every node in $V(G)$.
    Since $\mathcal{D}$ is minimal, \textsc{Impedensable-SDMDS}($i$) is false for every node $i$ in $\mathcal{D}$. Hence, the second action is disabled at every node $i$ in $\mathcal{D}$. 
    Thus, $\mathcal{D}$ remains unchanged.
\end{proof}

\subsection{Eventual Lattice-Linearity of \Cref{algorithm:rules-ds}}\label{subsection:ds-eventual}

\Cref{lemma:ds-addition} showed that \Cref{algorithm:rules-ds}.2 is lattice-linear. 
In this subsection, we make additional observations about \Cref{algorithm:rules-ds} to generalize the notion of lattice-linearity to eventually lattice-linear algorithms. 
We have the following observations.
\begin{enumerate}
    \item From \Cref{lemma:d-not-ds}, starting from any state, \Cref{algorithm:rules-ds} will reach a feasible state even if a node reads old information about the neighbours. This is due to the fact that \Cref{algorithm:rules-ds}.1 only adds nodes to $\mathcal{D}$. 
    If incorrect information about the state of neighbours causes $i$ not to be added to $\mathcal{D}$, this will be corrected when $i$ executes again and obtains recent information about neighbours. If incorrect information causes $i$ to be added to $\mathcal{D}$ unnecessarily, it does not affect this claim. 
    \item From \Cref{lemma:ds-addition}, if we start $G$ in a feasible state where no node has incorrect information about the neighbours in the initial state then \Cref{algorithm:rules-ds}.2 reaches a minimal dominating set. Note that this claim remains valid even if the nodes execute actions of \Cref{algorithm:rules-ds}.2 with old information about the neighbours as long as the initial information they use is correct. 
    \item We observe that \Cref{algorithm:rules-ds}.1 and \Cref{algorithm:rules-ds}.2 have very limited interference with each other, and so an arbitrary graph $G$ will reach an optimal state even if nodes are using old information. 
\end{enumerate}

From the above observations, if we allow the nodes to read old values, then the nodes can violate the feasibility of $G$ finitely many times and so $G$ will eventually reach a feasible state and stay there forever. 
We introduce the class of eventually lattice-linear algorithms (ELLA). \Cref{algorithm:rules-ds} is an ELLA.

\begin{definition}\label{definition:ella} \textbf{Eventually Lattice-Linear Algorithms (ELLA).}
    An algorithm $A$ is ELLA for a problem $P$, represented by a predicate $\mathcal{P}$, if its rules can be split into \arya{two sets of rules} $F_1$ and $F_2$ and there exists a subset $S_f$ of the state space $S$, such that 
    
    \begin{enumerate}[label=(\alph*)]
        \item \arya{Any computation of $A$ (from its permitted initial states)
        eventually reaches a state
        where $S_f$ is stable in $A$,
        i.e., $S_f$ is true and remains true subsequently. }
        \item Rules in $F_1$ are disabled in a state in $S_f$.
        \item \arya{$F_2$ is a lattice-linear algorithm, i.e., it follows \Cref{definition:ll-algos}, given that the system initializes in a state in $S_f$.}
    \end{enumerate}
\end{definition}

\arya{
\begin{definition}\label{definition:ellssa} \textbf{Eventually Lattice-Linear Self-Stabilizing (ELLSS) Algorithms.}
    Continuing from \Cref{definition:ella}, $A$ is an ELLSS algorithm iff $F_1$ takes the system to a state in $S_f$ from an arbitrary state, and $F_2$ is capable of taking the system from any state in $S_f$ to an optimal state.
\end{definition}

\noindent\textbf{\textit{Remark}}: The algorithms that we study in this paper are ELLSS algorithms, i.e., they follow \Cref{definition:ellssa}. Notice that \Cref{algorithm:rules-ds} is an ELLSS algorithm.
}

In \Cref{algorithm:rules-ds}, $F_1$ corresponds to \Cref{algorithm:rules-ds}.1 and $F_2$ corresponds to \Cref{algorithm:rules-ds}.2.
This algorithm satisfies the properties of \Cref{definition:ellssa}. 

\begin{example}\label{example:4-nodes}
    We illustrate the eventual lattice-linear structure of \Cref{algorithm:rules-ds} where we consider the special case where all nodes have the same single service and demand. Effectively, it becomes a case of minimal dominating set.

    In \Cref{figure:half-lattices-from-ds-example}, we consider an example of graph $G_4$ containing four nodes connected in such a way that they form two disjoint edges, i.e., $V(G_4)=\{v_1,v_2,v_3,v_4\}$ and $E(G_4)=\{\{v_1,v_2\},\{v_3,v_4\}\}$. 
    
    We write a state of this graph as $(v_1[st], v_2[st], v_3[st], v_4[st])$.
    As shown in this figure, of the 16 states in the state space, 9 are part of 4 disjoint lattices. These are feasible states, i.e., states where nodes with $st$ equals $IN$ form a (possibly non-minimal) dominating set. And, the remaining 7 are not part of any lattice. These are infeasible states, i.e., states where nodes with $st$ equals $IN$ do not form a dominating set. The states not taking part in any lattice structure (the infeasible states) are not shown in \Cref{figure:half-lattices-from-ds-example}.

    \arya{In a non-feasible state, some node will be addable. The instruction executed by addable nodes is not lattice-linear: an addable node moves in the dominating set unconditionally. After this, when no node is addable, then the global state $s$ becomes feasible state, i.e., $s$ manifests a valid dominating set. In $s$, however, some nodes may be removable. Only the removable nodes can be \imped. The instruction executed by an \imped node is lattice-linear.
    
    E.g., notice in \Cref{figure:half-lattices-from-ds-example} (a), assuming that the initial state is $(IN$, $IN$, $IN$, $IN)$, that $v_2$ and $v_4$ are \imped. Since they execute asynchronously, a lattice is induced among all possible global states that $G_4$ transitions through. If only $v_2$ (respectively, $v_4$) executes, the global state we obtain is $(IN$, $OUT$, $IN$, $IN)$ (respectively, $(IN$, $IN$, $IN$, $OUT)$). Since eventually both the nodes change their local states, we obtain the global state $(IN$, $OUT$, $IN$, $OUT)$.
    }
    \qed
\end{example}

\begin{figure}[ht]
    \centering
    \subfigure[]{
        \begin{tikzpicture}[scale=.7,every node/.style={scale=.7}]
            \node at (0,0) (a1) {(IN,OUT,IN,OUT)};
            \node at (-1.5,-1) (a2) {(IN,OUT,IN,IN)};
            \node at (1.5,-1) (a3) {(IN,IN,IN,OUT)};
            \node at (0,-2) (a4) {(IN,IN,IN,IN)};
            \draw (a1) -- (a2);
            \draw (a1) -- (a3);
            \draw (a2) -- (a4);
            \draw (a3) -- (a4);
        \end{tikzpicture}
    }\quad\quad 
    \subfigure[]{
        \begin{tikzpicture}[scale=.7,every node/.style={scale=.7}]
            \node at (0,0) (a1) {(OUT,IN,OUT,IN)};
            \node at (0,-1) (a2) {//only 1 state};
        \end{tikzpicture}
    }\\
    \subfigure[]{
        \begin{tikzpicture}[scale=.7,every node/.style={scale=.7}]
            \node at (0,0) (a1) {(OUT,IN,IN,OUT)};
            \node at (0,-1) (a2) {(OUT,IN,IN,IN)};
            \draw (a1) -- (a2);
        \end{tikzpicture}
    }\quad\quad 
    \subfigure[]{
        \begin{tikzpicture}[scale=.7,every node/.style={scale=.7}]
            \node at (0,0) (a1) {(IN,OUT,OUT,IN)};
            \node at (0,-1) (a2) {(IN,IN,OUT,IN)};
            \draw (a1) -- (a2);
        \end{tikzpicture}
    }
    \caption{Example lattice induced by \Cref{algorithm:rules-ds}.1 in $G_4$ ($G_4$ is described in \Cref{example:4-nodes}).}
    \label{figure:half-lattices-from-ds-example}
\end{figure}
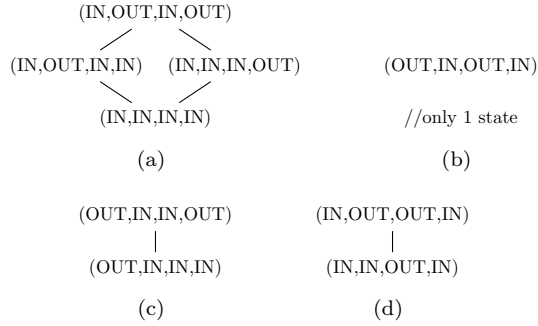

\subsection{Analysis  of \Cref{algorithm:rules-ds}: Time and Space complexity}\label{subsection:time-space-complexity-analysis}

\begin{theorem}\label{theorem:ds-convergence-time}
Starting from an arbitrary state, \Cref{algorithm:rules-ds} reaches an optimal state within $2n$ moves (or more precisely 1 round plus $n$ moves).
\end{theorem}

\begin{proof}
    From \Cref{lemma:d-not-ds}, we have that starting from an arbitrary state, \Cref{algorithm:rules-ds} will reach a feasible state within one round (or within $n$ moves).
    
    After that, if the input graph $G$ is not in an optimal state, then at least one node moves out such that $G$ stays in a feasible state (\Cref{lemma:ds-removal}). Thus, $G$ manifests an optimal state within $n$ additional moves.
\end{proof}

\begin{corollary}\label{corollary:algo-stabilizing-silent}
    \Cref{algorithm:rules-ds} is self-stabilizing and silent.
\end{corollary}

\begin{observation}
    At any time-step, a node will take $O((\Delta)^4\times (max_d)^2)$ time, where
    (1) $\Delta$ is the maximum degree of any node in $V(G)$, and (2) $max_d$ is the total number of distinct demands made by all the nodes in $V(G)$.
\end{observation}

\section{Applying ELLSS in Minimal Vertex Cover}\label{section:mvc}

The execution of \Cref{algorithm:rules-ds} was divided in two phases, (1) where the system reaches a feasible state (reduction of \textsc{Rank} to $0$), and (2) where the system reaches an optimal state (reduction of \textsc{Badness} to $0$).

Such design defines the concept of ELLSS algorithms. This design can be extended to numerous other problems 
where the optimal global state can be defined in terms of a minimal (or maximal) set $\mathcal{S}$ of nodes.
This includes the minimal vertex cover (\mvc) problem, maximal independent set problem and their variants. In this section, we discuss the extension to \mvc. 

In the \textit{minimal vertex cover} problem, the input is an arbitrary graph $G$, and the task is to compute a minimal set $\mathcal{V}$ such that for any edge $\{i,j\}\in E(G)$, $(i\in \mathcal{V})$ or $(j\in \mathcal{V})$. If a node $i$ is in $\mathcal{V}$, then $i[st]=IN$, otherwise $i[st]=OUT$.

The proposition $\mathcal{P}^\prime_v$ defining a feasible state and the proposition $\mathcal{P}_v$ defining the optimal state can be defined as follows.
\begin{center}
    $\mathcal{P}_v^\prime(\mathcal{V})\equiv \forall i\in V(G):((i\in \mathcal{V})\lor (\forall j\in Adj_i, j\in \mathcal{V}))$.\\
    $\mathcal{P}_v(\mathcal{V})\equiv \mathcal{P}_v^\prime(\mathcal{V}) \land (\forall i\in \mathcal{V}, \lnot\mathcal{P}_v^\prime(\mathcal{V}\setminus \{i\})).$
\end{center}

To develop an algorithm for \mvc, we utilize the macros in the following table. A node $i$ is \textit{removable} if all the nodes in its neighbourhood are in the vertex cover (\vc). $i$ is \textit{addable} if $i$ is not in the \vc and there is some node adjacent to it that is not in the \vc. $i$ is \textit{\imped} if  $i$ is in the \vc, and $i$ is the highest ID node that is removable in its distance-1 neighbourhood.

\begin{center}
    \begin{tabular}{|l|}
        \hline
        \textsc{Removable-\mvc}$(i)\equiv \forall j \in Adj_i, j[st]=IN$.\\
        \textsc{Addable-\mvc}$(i)\equiv i[st]=OUT\land(\exists j\in Adj_i:j[st]=OUT)$.\\
        \textsc{Impedensable-\mvc}$(i)\equiv i[st]=IN$ $\land$ \textsc{Removable-\mvc}$(i)\land$\\
        \quad\quad\quad\quad $(\forall j\in Adj_i: j[id]<i[id]\lor \lnot$\textsc{Removable-\mvc}$(j))$. \\
        \hline
    \end{tabular}
\end{center}

Based on the definitions above, the algorithm for \mvc is described as follows. If a node is addable, then it moves into the VC. If a node is \imped, then it moves out of the VC.
\begin{algorithm}\label{algorithm:rules-mvc}Rules for node $i$.
    \begin{center}
        \begin{tabular}{|l|}
            \hline
            \textsc{Addable-\mvc}$(i)\longrightarrow i[st]=IN$.\\
            \textsc{Impedensable-\mvc}$(i)\longrightarrow i[st]=OUT$.\\
            \hline
        \end{tabular}
    \end{center}
\end{algorithm}

\Cref{algorithm:rules-mvc} is an ELLSS algorithm in that it satisfies the conditions in \Cref{definition:ellssa}, where $F_1$ corresponds to the first action of \Cref{algorithm:rules-mvc}, $F_2$ corresponds to its second action, and $S_f$ is the set of the states for which $\mathcal{P}_v^\prime$ holds true.
Thus, starting from any arbitrary state, the algorithm eventually reaches a state where $\mathcal{V}$ is a minimal vertex cover.

\arya{
\begin{lemma}
    \Cref{algorithm:rules-mvc} is a silent eventually lattice-linear self-stabilizing algorithm for minimal vertex cover.
\end{lemma}

\begin{proof}
    In an arbitrary non-feasible state (where the input graph $G$ does not manifest a valid \vc), there is at least one node that is addable. An addable node immediately executes the first instruction of \Cref{algorithm:rules-mvc} and moves in the \vc. This implies that by the end of the first round, we obtain a valid (possibly non-minimal) \vc.

    If the input graph $G$ is in a feasible, but not optimal, state (where $G$ manifests a non-minimal \vc), then there is at least one removable node. This implies that there is at least one \imped node $i$ in that state (e.g., the removable node with the highest ID).
    Under \Cref{algorithm:rules-mvc}, any node in $Adj_i$ will not execute until $i$ changes its state. $i$ is removable because all nodes in $Adj_i$, along with $i$, are in the vertex cover. Thus $i$ must execute so that it becomes non-removable. This shows that  the second rule in \Cref{algorithm:rules-mvc} is lattice-linear.

    In a non-minimal, but valid, \vc, there is at least one node that is \imped, thus, with every move, the size of the vertex cover, manifested by $G$, reduces by 1. Also, notice that when an \imped node $i$ changes its state, no node in $Adj_i$ changes its state simultaneously. Thus, the validity of the vertex cover is not impacted when $i$ moves. Therefore, \Cref{algorithm:rules-mvc} is self-stabilizing.

    When $G$ manifests a minimal vertex cover, no node is addable or removable. This shows that \Cref{algorithm:rules-mvc} is silent.
\end{proof}
}

Observe that in \Cref{algorithm:rules-mvc}, the definition of \textsc{Impedensable} relies only on the information about distance-2 neighbours. Hence, the evaluation of guards take $O(\Delta^2)$ time. In contrast, (the standard) minimal dominating set problem would require the information of distance-4 neighbours to evaluate \textsc{Impedensable}. Hence, the evaluation of guards in that would take $O(\Delta^4)$ time. This algorithm converges in $2n$ moves (or more precisely 1 round plus $n$ moves).

\section{Applying ELLSS in Maximal Independent set}\label{section:mis}

In this section, we consider the application of ELLSS in the problem of maximal independent set (\mis). Unlike \mvc and SDMDS problems where we tried to reach a minimal set, here,  we have to obtain a maximal set. In the \textit{maximal independent set} problem, the input is an arbitrary graph $G$, and the task is to compute a maximal set $\mathcal{I}$ such that for any two nodes $i\in\mathcal{I}$ and $j\in\mathcal{I}$, if $i\neq j$, then $i$ and $j$ are not adjacent. If a node $i$ is in $\mathcal{I}$, then $i[st]=IN$, otherwise $i[st]=OUT$.

The proposition $\mathcal{P}^\prime_i$ defining a feasible state and the proposition $\mathcal{P}_i$ defining the optimal state can be defined as follows.
\begin{center}
    $\mathcal{P}_i^\prime(\mathcal{I})\equiv \forall i\in V(G):((i\not\in\mathcal{I})\lor (\forall j\in Adj_i: j\not\in \mathcal{I}))$.\\
    $\mathcal{P}_i(\mathcal{I})\equiv \mathcal{P}_i^\prime(\mathcal{I})\land(\forall i\in V(G): \lnot\mathcal{P}_i^\prime(\mathcal{I}\cup\{i\}))$.
\end{center}

To develop the algorithm for \mis, we define the macros in the following table. A node $i$ is \textit{addable} if all the neighbours of $i$ are out of the independent set(\is). A node is \textit{removable} if $i$ is in the \is and there is some neighbour of $i$ that is also in \is. $i$ is \textit{\imped} if $i$ is out of the IS, and $i$ is the highest ID node in its distance-1 neighbourhood that is addable.

\begin{center}
    \begin{tabular}{|l|}
        \hline
        \textsc{Addable}($i)\equiv \forall j \in Adj_i, j[st]=OUT$.\\
        \textsc{Removable-\mis}$(i)\equiv i[st]=IN\land(\exists j\in Adj_i:j[st]=IN$).\\
        \textsc{Impedensable-\mis}$(i)\equiv i[st]=OUT\land$ \textsc{Addable}$(i)\land$\\
        \quad\quad\quad\quad $(\forall j\in Adj_i:j[id]<i[id]\lor\lnot$\textsc{Addable}$(j))$.\\
        \hline
    \end{tabular}
\end{center}

Based on the definitions above, the algorithm for \mis is described as follows. If a node $i$ is \imped, then it moves into the \is. If $i$ is removable, then it moves out of the \is.

\begin{algorithm}\label{algorithm:rules-mis}Rules for node $i$.
    \begin{center}
        \begin{tabular}{|l|}
            \hline
            \textsc{Removable-\mis}$(i)\longrightarrow i[st]=OUT$.\\
            \textsc{Impedensable-\mis}$(i)\longrightarrow i[st]=IN$.\\
            \hline
        \end{tabular}
    \end{center}
\end{algorithm}

This algorithm is an ELLSS algorithm as well: as per \Cref{definition:ellssa}, $F_1$ corresponds to the first action of \Cref{algorithm:rules-mvc}, $F_2$ corresponds to its second action, and and $S_f$ is the set of the states for which $\mathcal{P}_i^\prime$ holds true.
Thus, starting from any arbitrary state, the algorithm eventually reaches a state where $\mathcal{I}$ is a maximal independent set.

\arya{
\begin{lemma}
    \Cref{algorithm:rules-mis} is a silent eventually lattice-linear self-stabilizing algorithm for maximal independent set.
\end{lemma}

\begin{proof}
    In an arbitrary non-feasible state (where the input graph $G$ does not manifest a valid \is), there is at least one node that is removable. A removable node immediately executes the first instruction of \Cref{algorithm:rules-mis} and moves out of the \is. This implies that by the end of the first round, we obtain a valid (possibly non-minimal) \is.

    If the input graph $G$ is in a feasible, but not optimal, state (where $G$ manifests a non-minimal \is), then there is at least one addable node. This implies that there is at least one \imped node $i$ in that state (e.g., the addable node with the highest ID).
    Under \Cref{algorithm:rules-mvc}, any node in $Adj_i$ will not execute until $i$ changes its state. $i$ is addable because all nodes in $Adj_i$, along with $i$, are out of the independent set. Thus $i$ must execute so that it becomes non-addable. This shows that  the second rule in \Cref{algorithm:rules-mis} is lattice-linear.

    Since in a non-minimal, but valid, independent set, there is at least one node that is \imped, we have that with every move, the size of the independent set, manifested by $G$, reduces by 1. Also, notice that when an \imped node $i$ changes its state, no node in $Adj_i$ changes its state simultaneously. Thus, the validity of the independent set is not impacted when $i$ moves. Therefore, we have that \Cref{algorithm:rules-mis} is self-stabilizing.

    When $G$ manifests a maximal independent set, no node is removable or addable. This shows that \Cref{algorithm:rules-mis} is silent.
\end{proof}
}

In \Cref{algorithm:rules-mis}, the definition of \textsc{Addable} relies only on the information about distance-2 neighbours. Hence, the evaluation of guards take $O(\Delta^2)$ time. This algorithm converges in $2n$ moves (or more precisely 1 round plus $n$ moves).

\section{Applying ELLSS in Colouring}\label{section:gc}

In this section, we extend ELLSS algorithms to graph colouring. 
In the \textit{graph colouring} (GC) problem, the input is a graph $G$ and the task is to (re-)assign colours to all the nodes such that no two adjacent nodes have the same colour.

Unlike \mvc, \mds or \mis, colouring does not have a binary domain. Instead, we correspond the equivalence of changing the state to $IN$ to the case where a node sets its colour to $i[id]+n$. And, the equivalence of changing the state to $OUT$ corresponds to the case where a node decreases its colour.

The proposition $\mathcal{P}^\prime_c$ defining a feasible state and the proposition $\mathcal{P}_c$ defining an optimal state is defined below. $\mathcal{P}_c$ is true when all the nodes have lowest available colour, that is, for any node $i$ and for all colours $c$ in $[1:i[colour]-1]$, $c$ equals the colour of one of the neighbours $j$ of $i$.

\begin{center}
    $\mathcal{P}_c^\prime(G)\equiv \forall i\in V(G),\forall j\in Adj_i:i[colour]\neq j[colour]$.\\
    $\mathcal{P}_c(G)\equiv \mathcal{P}_c^\prime\land (\forall i\in V(G):(\forall c\in [1:i[colour]-1]:(\exists j\in Adj_i: j[colour]=c)))$.
\end{center}

We define the macros as shown in the following table. A node $i$ is \textit{conflicted} if it has a conflicting colour with at least one of its neighbours. $i$ is subtractable if there is a colour value less than $i[colour]$ that $i$ can change to without a conflict with any of its neighbours. $i$ is \textit{\imped} if $i$ is not conflicted, and it is the highest ID node that is subtractable.
\begin{center}
    \begin{tabular}{|l|}
        \hline
        \textsc{Conflicted}($i)\equiv \exists j \in Adj_i:j[colour]=i[colour]$.\\
        \textsc{Subtractable}($i)\equiv \exists c\in [1:i[colour]-1]: \forall j\in Adj_i: j[colour]\neq c$.\\
        \textsc{Impedensable-GC}$(i)\equiv\lnot$\textsc{Conflicted}($i$) $\land$ \textsc{Subtractable}$(i)\land$\\
        \quad\quad $(\forall j\in V(G):\lnot\textsc{Conflicted}(j)\land(j[id]<i[id]\lor \lnot$\textsc{Subtractable}$(j)))$.\\
        \hline
    \end{tabular} 
\end{center}

Unlike SDMDS, \mvc and \mis, in graph colouring (GC), each node is associated with a variable $colour$ that can take several possible values (the domain can be as large as the set of natural numbers). 
As mentioned above, the action of setting a colour value to $i[colour]+i[id]$ is done whenever a conflict is detected. 
Effectively, this is like setting the colour to an error value such that the error value of every node is distinct in order to avoid a conflict. This error value will be reduced when node $i$ becomes \imped and decreases its colour.

The actions of the algorithm are shown in \Cref{algorithm:rules-c}. If a node $i$ is \imped, then it changes its colour to the minimum possible colour value. If $i$ is conflicted, then it changes its colour value to $i[colour]+i[id]$.

\begin{algorithm}\label{algorithm:rules-c}Rules for node $i$.
    \begin{center}
        \begin{tabular}{|l|}
            \hline
            \textsc{Conflicted-GC}$(i)$ $\longrightarrow i[colour]=i[colour]+i[id]$.\\
            \textsc{Impedensable-GC}$(i)\longrightarrow$\\
            \quad\quad $i[colour]=\min\limits_{c}\{c\in [1:i[colour]-1]:(\forall j\in Adj_i: j[colour]\neq c)\}$.\\
            \hline
        \end{tabular}
    \end{center}
\end{algorithm}

\Cref{algorithm:rules-c} is an ELLSS algorithm: according to \Cref{definition:ellssa}, $F_1$ corresponds to the first action of \Cref{algorithm:rules-mvc}, $F_2$ corresponds to its second action, and and $S_f$ is the set of the states for which $\mathcal{P}_c^\prime$ holds true.
Thus, starting from any arbitrary state, the algorithm eventually reaches a state where no two adjacent nodes have the same colour and no node can reduce its colour.

\arya{
\begin{lemma}
    \Cref{algorithm:rules-c} is a silent eventually lattice-linear self-stabilizing algorithm for graph colouring.
\end{lemma}

\begin{proof}
    In an arbitrary non-feasible state (where the input graph $G$ does not manifest a valid colouring), there is at least one node that is conflicted. A conflicted node immediately executes the first instruction of \Cref{algorithm:rules-c} and makes its colour equal to its $ID$ plus its colour value. Since the value $i[colour]+i[id]$ by which a node updates its colour value will resolve such conflict with one adjacent node in 1 move, $i$ will become non-conflicted in almost $deg(i)$ moves.

    If the input graph $G$ is in a feasible, but not optimal, state (where $G$ manifests a valid colouring but some nodes can reduce their colour), then there is at least one subtractable node. This implies that there is an \imped node $i$ in that state (the subtractable node with the highest ID).
    Under \Cref{algorithm:rules-c}, any node will not execute until $i$ changes its state. $i$ is subtractable because there is a colour value $c$ less that $i[colour]$ such that no node in $Adj_i$ has that colour value. Thus $i$ must execute to become non-subtractable. This shows that the second rule in \Cref{algorithm:rules-c} is lattice-linear.

    Since in a non-minimal, but valid, colouring, there is at least one node $i$ that is \imped, we have that a node will become non-subtractable in atmost $deg(i)$ moves. Notice that when an \imped node $i$ changes its state, no node changes its state simultaneously. Also, the reduced colour will not have a conflict with any other node. Thus, no conflicts arise. Therefore, we have that \Cref{algorithm:rules-c} is self-stabilizing.

    When $G$ manifests a valid non-subtractable colouring, no node is removable or addable. This shows that \Cref{algorithm:rules-c} is silent.
\end{proof}
}

In \Cref{algorithm:rules-c}, the definition of \textsc{\imped} relies only on the information about distance-2 neighbours. Hence, the evaluation of guards take $O(n)$ time. This algorithm converges in $2m + (n+2m)=n+4m$ moves.

\section{Applying ELLSS in 2-Dominating Set problem}\label{section:2ds}

The 2-dominating set (2DS) problem 
provides a stronger form of dominating set (DS), as compared to the usual \mds problem. In the \textit{2-dominatind set} problem, the input is a graph $G$ with nodes having domain $\{IN,OUT\}$. The task is to compute a set $\mathcal{D}$ where some node $i\in \mathcal{D}$ iff $i[st]=IN$; $\mathcal{D}$ must be computed such that there are no two nodes $j,k\in V(G)$ that are in $\mathcal{D}$, and a node $i\in V(G)$ that is not in $\mathcal{D}$, such that $\mathcal{D}\cup\{i\}\setminus \{j,k\}$ is a valid \ds.

Unlike the SDMDS, \mvc, \mis or GC problems that simply study the condition of their immediate neighbours before they change their state, and after they would change their state, the 2-DS problem looks one step further.
Specifically, the usual \mds or \mvc problems investigate the computation of any minimal DS or VC respectively, whereas the 2DS problem requires the computation of such a DS where it must not be the case that another valid DS can be computed while removing two nodes from it and adding one node to it.

The propositions $\mathcal{P}^\prime_{d}$ defines a \ds, $\mathcal{P}_d$ defines an \mds and $\mathcal{P}_{2d}$ defines an optimal state, obtaining a \tds. These propositions are defined below.

\begin{center}
    $\mathcal{P}_{d}^\prime(\mathcal{D})$
    $\equiv \forall i\in V(G):i\in\mathcal{D}\lor (\exists j\in Adj_i:j\in\mathcal{D})$.\\
    $\mathcal{P}_{d}(\mathcal{D})$
    $\equiv\mathcal{P}_{d}'(\mathcal{D})\land(\forall(i\in V(G):\lnot\mathcal{P}_d(D\setminus \{i\})))$.\\
    $\mathcal{P}_{2d}(\mathcal{D})$
    $\equiv \mathcal{P}_{d}(\mathcal{D})  
    \land \neg (\exists i\in V(G),i\not\in\mathcal{D}:$\\$(\exists j, k\in Adj_i, j\in\mathcal{D},k\in\mathcal{D} : \mathcal{P}'_{d}(\mathcal{D} \cup \{i\}\setminus\{j, k\})))$ 
\end{center}

\arya{Our algorithm is based on the following intuition: Let $\mathcal{D}$ be an \mds. If there exists nodes $i, j$ and $k$ such that $j, k \in \mathcal{D}$ and $i \not \in \mathcal{D}$, and $\mathcal{D} \cup \{i\} - \{j, k\}$ is also a \ds, then $j$ and $k$ must be neighbors of $i$.}

The macros that we utilize are as follows. A node $i$ is \textit{addable} if $i[st]=OUT$ and all the neighbours of $i$ are also out of the DS. $i$ is \textit{removable} if $i[st]=IN$ and there exists at least one neighbour of $i$ that is also in the DS. A node $i$ is \textit{2-addable} if $i[st]=OUT$ there exist nodes $j$ and $k$ in the distance-2 neighbourhood of $i$ where $j[st]=IN$ and $k[st]=IN$ such that $j$ and $k$ can be removed and $i$ can be added to the DS such that $j,k$ and their neighbours stay dominated. A node is \textit{unsatisfied} if it is removable or 2-addable. A node is \textit{\imped} if it is the highest id node in its distance-4 neighbourhood that is unsatisfied.

\begin{center}
    \begin{tabular}{|l|}
        \hline 
        $\textsc{Addable-2DS}\equiv i[st]=OUT\land (\forall j\in Adj_i:j[st]=OUT)$.\\
        $\textsc{Removable-2DS}(i)\equiv i[st]=IN\land (\forall j\in Adj_i\cup\{i\}:((j\neq i\land j[st]=IN)$\\
        \quad\quad\quad\quad $\lor$ $(\exists k\in Adj_j, k\neq i: k[st]=IN)))$.\\
        $\textsc{Two-Addable-2DS}(i) \equiv i[st]=OUT\land(\forall j\in Adj^2_i\cup\{i\}:$\\
        \quad\quad\quad\quad $\lnot (\textsc{Addable-2DS}(j)\lor\textsc{Removable-2DS}(j))) \land$\\
        \quad\quad\quad\quad $(\exists j,k\in Adj^2_i, j[st]=IN, k[st]=IN:$\\
        \quad\quad\quad\quad $(\forall q\in Adj_j\cup Adj_k\cup\{j,k\}:(\exists r\in Adj_q:r[st]=IN\lor r=i)))$.\\
        $\textsc{Unsatisfied-2DS}(i)\equiv \textsc{Removable-2DS}(i)\lor \textsc{Two-Addable-2DS}(i)$.\\
        $\textsc{\Imped-2DS}(i) \equiv \textsc{Unsatisfied-DS}(i) \land(\forall j\in Adj^4_i:$\\
        \quad\quad\quad\quad $(\lnot \textsc{Unsatisfied-DS}(j) \lor i[id]>j[id]))$.\\
        \hline 
    \end{tabular}
\end{center}

The algorithm for the 2-dominating set problem is as follows.
If a node $i$ is addable, then it turns itself in the DS, ensuring that $i$ and all it neighbouring nodes it stay dominated. 
As stated above, a node is \imped then it is either removable or 2-addable.
If a node is \imped and removable, then it turns itself out of the DS, ensuring that $i$ is not such a node that is not needed in the DS, but is still present in the DS. If $i$ is \imped and 2-addable, then there are two nodes $j$ and $k$ in the DS such that $j$ and $k$ can be removed, and $i$ can be added, and the resulting DS is still a valid DS. In this case, $i$ moves into the DS, and moves $j$ and $k$ out of the DS.

\begin{algorithm}\label{algorithm:rules-2ds-v1}Rules for node $i$.
$$
    \begin{array}{|l|}
        \hline 
        \textsc{Addable-2DS}(i)\longrightarrow i[st]= IN.\\
        \textsc{\Imped-2DS}(i)\longrightarrow\\
        \begin{cases}
            i[st]=OUT. & \text{if $i[st]=IN$.}\\
            {j[st]=OUT, k[st]=OUT, i[st]=IN.} & \text{if $i[st]=OUT$.}
        \end{cases}~\\
        // \text{The reference to $j$ and $k$ is from the definition of $\textsc{Two-Addable}(i)$}\\
        \hline 
    \end{array}
$$
\end{algorithm}

This is an ELLSS algorithm that works in three phases: first, every node $i$ checks if it is addable. If $i$ is not addable, then $i$ checks if it is \imped and removable, providing a minimal DS. And finally, $i$ checks if it is \imped and 2-addable, providing a 2DS.
Thus, this algorithm satisfies the conditions in \Cref{definition:ellssa}, where $F_1$ constitutes of the first action of \Cref{algorithm:rules-2ds-v1}, $F_2$ corresponds to its second action, and $S_f$ is the set of the states for which $\mathcal{P}_{d}^\prime$ holds true.
Thus, starting from any arbitrary state, the algorithm eventually reaches a state where $\mathcal{D}$ is 2-dominating set.

\arya{
\begin{lemma}
    \Cref{algorithm:rules-mvc} is a silent eventually lattice-linear self-stabilizing algorithm for 2-dominating set.
\end{lemma}

\begin{proof}
    In an arbitrary non-feasible state (where the input graph $G$ does not manifest a valid \ds), there is at least one node that is addable. An addable node immediately executes the first instruction of \Cref{algorithm:rules-2ds-v1} and moves in the \ds. This implies that by the end of the first round, we obtain a valid (possibly non-minimal) \ds.
    
    If the input graph $G$ is in a feasible, but not optimal, state (where $G$ manifests a non-minimal \ds), then there is at least one node that is removable or 2-addable. This implies that there is at least one \imped node $i$ in that state (e.g., the node, which is removable or 2-addable, with the highest ID).
    Under \Cref{algorithm:rules-mvc}, any node in $Adj^4_i$ will not execute until $i$ changes its state. If $i$ is removable, then all its neighbours are being dominated by a node other than $i$. If $i$ is 2-addable, then there exists a pair of nodes $j$ and $k$ such that if $j$ and $k$ can move out and $i$ moves in, then all nodes in $Adj_i$, $Adj_j$ and $Adj_k$ will stay dominated, including $i$, $j$ and $k$. Thus $i$ must execute so that it becomes non-\imped. This shows that  the second rule in \Cref{algorithm:rules-2ds-v1} is lattice-linear.

    Notice that if an arbitrary node $j$ and $k$ can move out of the \ds given that all nodes stay dominated if $i$ moves in, then $j$ and $k$ must be the neighbours of $i$. This is assuming that $G$ is in a valid dominating set. Otherwise, it cannot be guaranteed that $i$ can dominate the nodes that only $j$ or $k$ are dominating.
    
    Since in a non-minimal, but valid, \ds, there is at least one node that is removable \imped, we have that with every move of a removable \imped node, the size of the \ds, manifested by $G$, reduces by 1.
    Now assume that $G$ manifests a \ds such that no node is addable or removable. Here, if $G$ does not manifest a 2-dominating set, then, from the discussion from the above paragraph there must exist at least one set of three nodes $i$, $j$ and $k$ such that $j$ and $k$ can move out and $i$ can move in guaranteeing that all nodes in $Adj_i$, $Adj_j$ and $Adj_k$ stay dominated, including $i$, $j$ and $k$. With every move of a 2-addable \imped node, the size of the \ds, manifested by $G$, reduces by 1.
    Also, notice that when an \imped node $i$ changes its state, no node in $Adj^4_i$ changes its state simultaneously. Thus, the validity of the \ds is not impacted when $i$ moves. Therefore, we have that \Cref{algorithm:rules-2ds-v1} is self-stabilizing.
    
    When $G$ manifests a 2-dominating set, no node is addable, removable or 2-addable. This shows that \Cref{algorithm:rules-2ds-v1} is silent.
\end{proof}
}

Note that in \Cref{algorithm:rules-2ds-v1}, the definition of \textsc{Removable} relies on the information about distance-2 neighbours, and consequently, the definition of \textsc{Two-Addable} relies on the information about distance-4 neighbours. Hence, because of the time complexity of evaluating if a node is \imped, the guards take $O(\Delta^8)$ time. 
This algorithm converges in $3n$ moves (or more precisely 1 round plus $2n$ moves).

In this algorithm, one of the actions is changing the states of 3 processes at once. 
However, it can be implemented in a way that a process changes its own state only. We sketch how this can be done as follows.
To require that a process only changes its own state, we will need additional variables so processes know that they are in the midst of an update where $i$ needs to add itself to $\mathcal{D}$ and $j$ and $k$ need to remove themselves from $\mathcal{D}$.
Intuitively, it will need a variable of the form $getout.i$ which will be set to $\{j, k\}$ to instruct $j$ and $k$ to leave the dominating set. 
When $j$ or $k$ are in the midst of leaving the dominating set, all the nodes in $Adj_i^6$ will have to wait until the operation is completed.
With this change, we note that the algorithm will not be able to tolerate incorrect initialization of $getout.i$ while preserving lattice-linearity. 

\section{Related Work}\label{section:literature}

\noindent \textbf{Lattice-Linearity}: Garg (2020) \cite{Garg2020} studied the exploitation of lattice-linear predicates in several problems to develop parallel processing algorithms. Lattice-linearity ensures convergence of the system to an optimal state while the nodes perform executions asynchronously, and are allowed to perform executions based on the old values of other nodes. Problems like the stable marriage problem, job scheduling, market clearing price and others are studied in \cite{Garg2020}. In \cite{Garg2021} and \cite{Garg2022}, the authors have studied lattice-linearity in, respectively, housing market problem and several dynamic programming problems. These papers study problems which possess a predicate (called a lattice-linear predicate) under which the global states form a lattice, and where the system needs to be initialized in a specific initial state, and hence does not support self-stabilization.

The problems that we study in this paper are the problems which do not possess any predicate under which the global states form a lattice. In addition the algorithms that we present are self-stabilizing and thus converge to an optimal state from an arbitrary state.

\noindent \textbf{Minimal dominating set and its generalizations}: Self-stabilizing algorithms for the minimal dominating set (\mds) problem, and its variations, have been proposed in several works in the literature, for example, in \cite{Xu2003, Hedetniemi2003, Turau2007, GODDARD2008, Chiu2014, Fink1985, Kobayashi2017}. The best convergence time among these works is $4n$ moves.

The minimal $k$-dominating set problem was studied in \cite{Fink1985}; here, the task is to compute a minimal set of nodes $\mathcal{D}$ such that for each node $v\in V(G)$, $v\in\mathcal{D}$ or there are at least $k$ neighbours of $v$ in $\mathcal{D}$. A generalization of the \mds is described in \cite{Kobayashi2022}, where the input includes wish sets for every node. For each node $i$, $i$ should be in the dominating set $\mathcal{D}$ or one of its wish set must be a subset of $\mathcal{D}$. This algorithm converges in $O(n^3m)$ steps and $O(n)$ rounds; time complexity of the evaluation of guards is exponential in the degree of the nodes.

The problem that we study is the service demand based minimal dominating set problem, which is a more practical generalization of \mds. The algorithm that we propose is self-stabilizing, converges in 1 round plus $n$ moves (within $2n$ moves), and does not require a synchronous environment, which is an improvement over the existing literature.
In addition, evaluation of guards takes only $O(\Delta^4)$ time, which is better than the algorithm presented in \cite{Kobayashi2022}.

\noindent \textbf{Minimal vertex cover}: Self-stabilizing algorithms for the minimal vertex cover problem
have been studied in Astrand and Suomela (2010) \cite{Astrand2010} that converges in $O(\Delta)$ rounds, and Turau (2010) \cite{TURAU2010} that converges in $O(\min\{n$, $\Delta^2$, $\Delta\log_3 n\})$ rounds.

The algorithm that we propose is self-stabilizing, converges in 1 round plus $n$ moves (within $2n$ moves), and does not require a synchronous environment.

\noindent \textbf{Maximal Independent Set}: Self-stabilizing algorithm for maximal independent set has been presented in \cite{Turau2007}, that converges in $\max\{3n-5, 2n\}$ moves under an unfair distributed scheduler, \cite{GODDARD2008} that converges in $n$ rounds under a distributed or synchronous scheduler, \cite{Hedetniemi2003} that converges in $2n$ moves.

The algorithm that we propose is self-stabilizing, converges in 1 round plus $n$ moves (within $2n$ moves), and does not require a synchronous environment.

\noindent\textbf{Colouring}: Self-stabilizing algorithms for decentralized (where nodes only read from their immediate neighbours) graph colouring have been presented in \cite{Bhartia2016,Checco2017,Duffy2013,Duffy2008,Galan2017,Leith2006,Motskin2009,Chakrabarty2020}. The best convergence time among these algorithms is $n\times \Delta$ moves, where $\Delta$ is the maximum degree of the input graph.

The algorithm that we propose is self-stabilizing, converges in $n+4m$ moves, and does not require a synchronous environment.

\textbf{A survey}: A survey of self-stabilizing algorithms on independence, domination and colouring problems can be found in \cite{Guellati2010}.

\noindent \textbf{2-Dominating set}: The 2-dominating set is not an extensively studied problem. The problem was introduced in \cite{Bollobas1990}. A self-stabilizing algorithm for the 2-dominating set problem has been studied in \cite{Maruyama2022}. This algorithm converges in $O(nD)$ rounds under a distributed scheduler, where $D$ is the diameter of $G$.

The algorithm present in this paper is self-stabilizing, converges in 1 round plus $2n$ moves (within $3n$ moves), and is tolerant to asynchrony.

\section{Experiments}\label{section:experiments}

In this section, we present the experimental results of time convergence of shared memory programs. We focus on the problem of maximal independent set (\Cref{algorithm:rules-mis}) as an example. 

\begin{figure}[ht]
    \centering
    \subfigure[]{
        \includegraphics[width=0.47\textwidth]{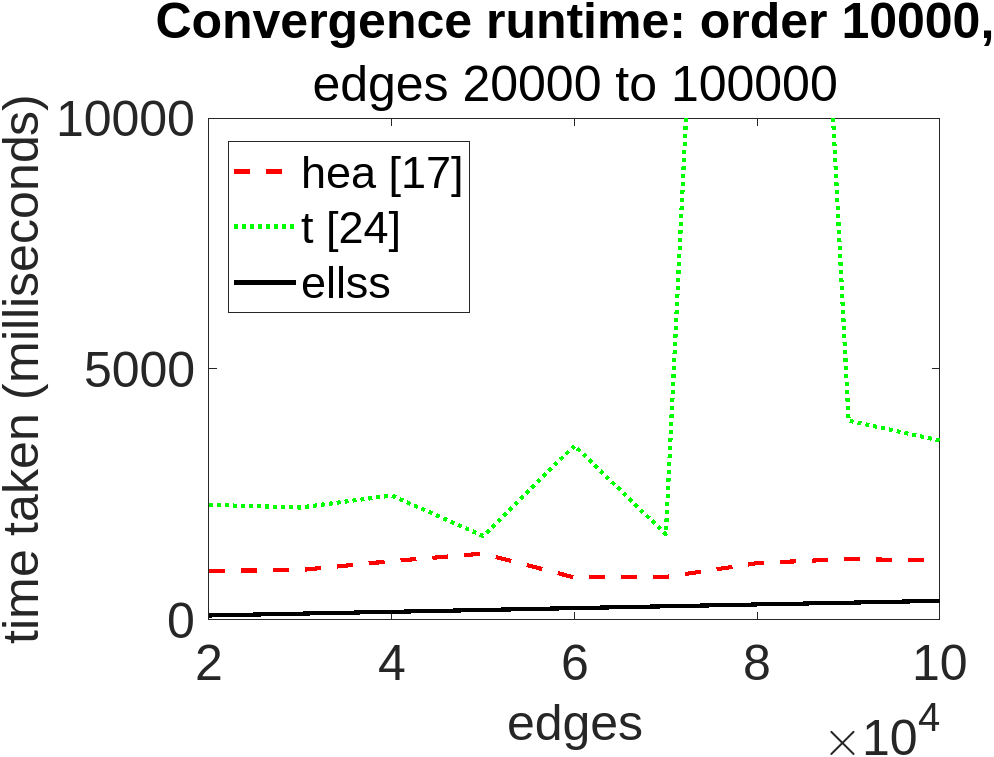}
    }
    \subfigure[]{
        \includegraphics[width=0.47\textwidth]{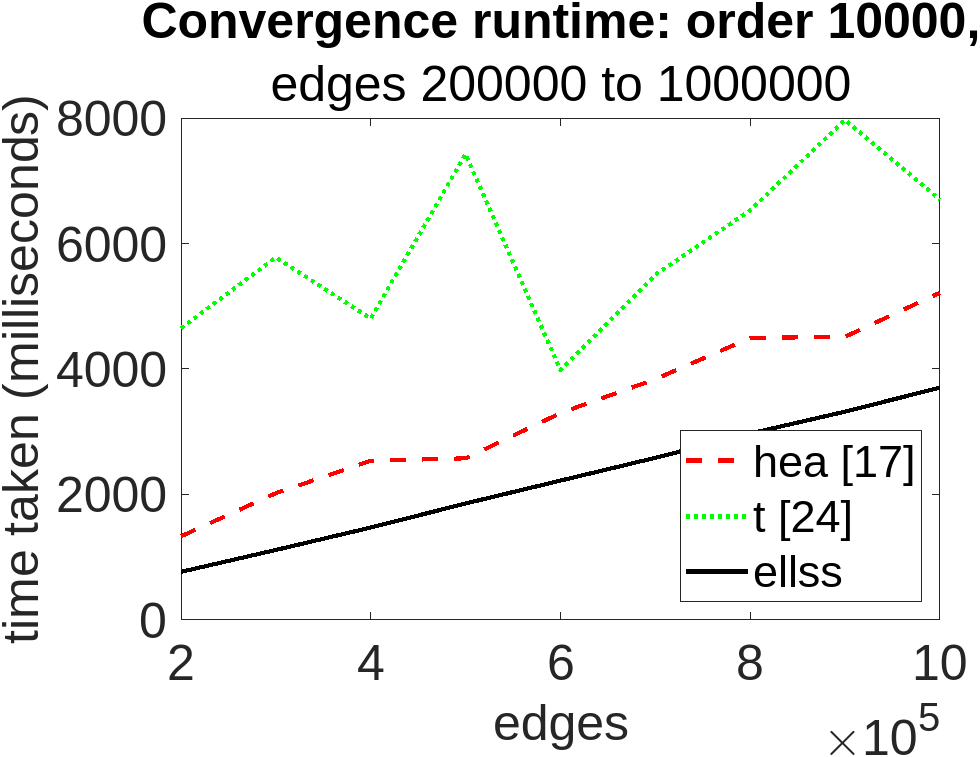}
    }
    \subfigure[]{
        \includegraphics[width=0.47\textwidth]{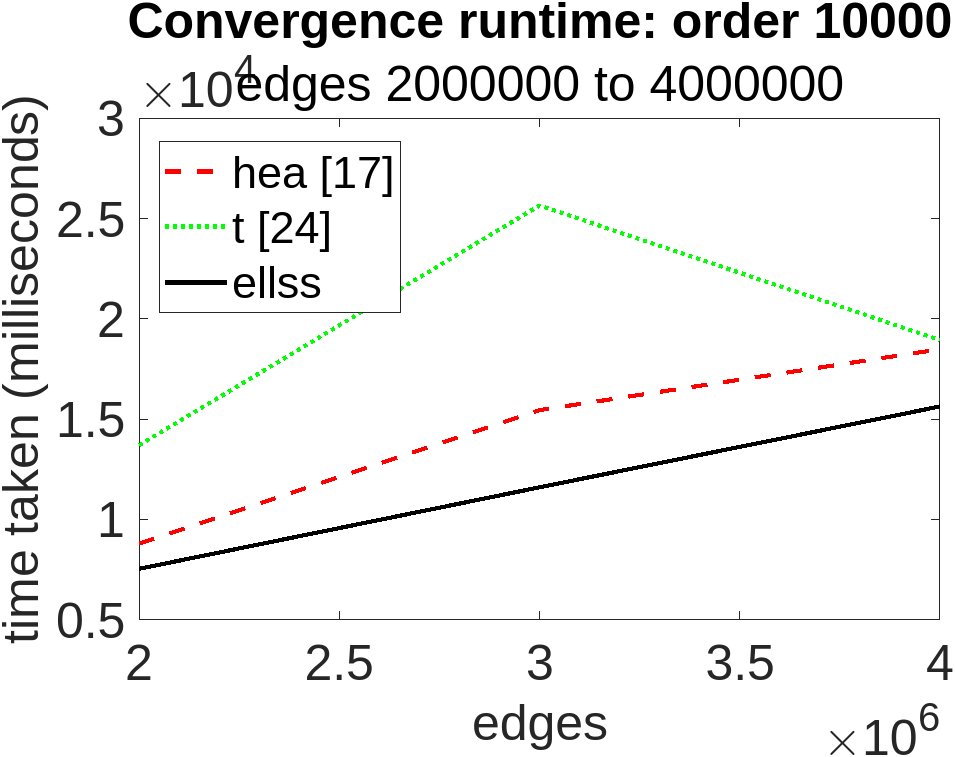}
    }
    \subfigure[]{
        \includegraphics[width=0.47\textwidth]{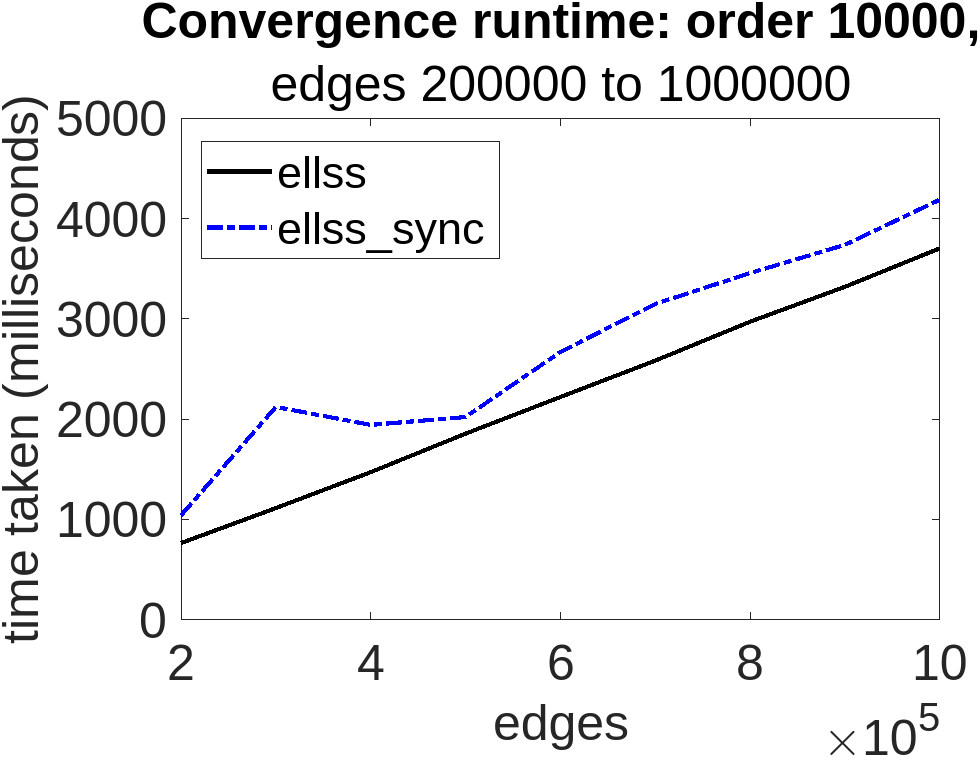}
    }
    \caption{Maximal Independent set algorithms convergence time on random graphs generated by \texttt{networkx} library of \texttt{python3}. All graphs are of 10,000 nodes. Comparision between runtime of \Cref{algorithm:rules-mis}, Hedetniemi et al. (2003) \cite{Hedetniemi2003} and Turau (2007) \cite{Turau2007} and synchronized \Cref{algorithm:rules-mis}. (a) 20,000 to 100,000 edges, \Cref{algorithm:rules-mis}, \cite{Hedetniemi2003} and \cite{Turau2007}. (b) 200,000 to 1,000,000 edges, \Cref{algorithm:rules-mis}, \cite{Hedetniemi2003} and \cite{Turau2007}. (c) 2,000,000 to 4,000,000 edges, \Cref{algorithm:rules-mis}, \cite{Hedetniemi2003} and \cite{Turau2007}. (d) 20,000 to 100,000 edges, \Cref{algorithm:rules-mis} and \Cref{algorithm:rules-mis} lockstep synchronized.}
    \label{figure:figure}
\end{figure}

We compare \Cref{algorithm:rules-mis} with the algorithms present in the literature for the maximal independent set problem. Specifically, we implemented the algorithms present in Hedetniemi et al. (2003) \cite{Hedetniemi2003} and Turau (2007) \cite{Turau2007}, and compare their convergence time. 
The input graphs were random graphs of order 10,000 nodes, generated by the \texttt{networkx} library of python. For comparing the performance results, all algorithms are run on the same set of graphs. 

The experiments are run on Cuda using the \texttt{gcccuda2019b} compiler. The program for \Cref{algorithm:rules-mis} was run asynchronously, and the programs in 
\cite{Hedetniemi2003} and \cite{Turau2007} are run under the required synchronization model.
The experiments are run on \texttt{Intel(R) Xeon(R) Platinum 8260 CPU @} \texttt{2.40} \texttt{GHz, cuda  V100S}. The programs are run using the command \texttt{nvcc $\langle$program$\rangle$.cu -G}. Here, each multiprocessor ran 256 threads. And, the system provided sufficient multiprocessors so that each node in the graph can have its own thread. All the observations are an average of 16 readings.

\Cref{figure:figure} (a) (respectively, \Cref{figure:figure} (b) and \Cref{figure:figure} (c)) shows a line graph comparision of the convergence time for these algorithms with the number of edges varying from 20,000 to 100,000  (respectively, 200,000 to 1,000,000 and 2,000,000 to 4,000,000). So, the average degree is varying from 4 to 20 (respectively, 40 to 200 and 400 to 800).
Observe that
the convergence time taken by the program for \Cref{algorithm:rules-mis} is significantly lower than the other two algorithms.

Next, we considered how much of the benefit of \Cref{algorithm:rules-mis} can be allocated to asynchrony due to the property of lattice-linearity. For this, we compared the performance of \Cref{algorithm:rules-mis} running in asynchrony (to allow nodes to read old/inconsistent values) and running in lock-step (to ensure that they only reads the most recent values). \Cref{figure:figure} (d) compares these results. We observe that the asynchronous implementation has lower convergence time.

We have performed the experiments on shared memory architecture that allows nodes to access all memory \textit{quickly}. This means that the overhead of synchronization is low. By contrast, if we had used a distributed system, where computing processors are far apart, the cost of synchronization will be even higher. Hence, the benefit of lattice-linearity (where synchronization is not needed) will be even higher.

\section{Conclusion}\label{section:conclusion}

We extended lattice-linearity from \cite{Garg2020} to the context of self-stabilizing algorithms. A key benefit of lattice-linear systems is that correctness is preserved even if nodes read old information about other nodes. However, the approach in \cite{Garg2020} relies on the assumption that the algorithm starts in  specific initial states, hence, it is not directly applicable in self-stabilizing algorithms. 

We began with the service demand based minimal dominating set (SDMDS) problem and designed a self-stabilizing algorithm for the same. Subsequently, we observed that it consists of two parts: One part makes sure that it gets the system to a state in $S_f$. The second part is a lattice-linear algorithm that constructs a minimal dominating set if it starts in some valid initial states, say a state in $S_f$. We showed that these parts have bounded interference, thus, they guarantee that the system stabilizes even if the nodes execute asynchronously. 

We defined the general structure of eventually lattice-linear self-stabilization to capture such algorithms. We demonstrated that it is possible to develop eventually lattice-linear self-stabilizing (ELLSS) algorithms for minimal vertex cover, maximal independent set, graph colouring and 2-dominating set problems. 

We also demonstrated that these algorithms substantially benefit from their ELLSS property. They outperform existing algorithms while they guarantee convergence without synchronization among processes.

Finally, as future work, an interesting direction can be to study which class of problems can the paradigm of ELLSS algorithms be extended to. Also, it is interesting to study if approximation algorithms for NP-Hard problems can be developed.

\bibliography{ella.bib}

\begin{thebibliography}{10}

\bibitem{Bhartia2016}
{\sc Bhartia, A., Chakrabarty, D., Chintalapudi, K., Qiu, L., Radunovic, B.,
  and Ramjee, R.}
\newblock Iq-hopping: Distributed oblivious channel selection for wireless
  networks.
\newblock In {\em Proceedings of the 17th ACM International Symposium on Mobile
  Ad Hoc Networking and Computing\/} (New York, NY, USA, 2016), MobiHoc '16,
  Association for Computing Machinery, p.~81–90.

\bibitem{Bollobas1990}
{\sc Bollobas, B., Cockayne, E.~J., and Mynhardt, C.~M.}
\newblock On generalised minimal domination parameters for paths.
\newblock {\em Discrete Mathematics 86\/} (1990), 89--97.

\bibitem{Chakrabarty2020}
{\sc Chakrabarty, D., and de~Supinski, P.}
\newblock {\em On a Decentralized $(\Delta+1)$-Graph Coloring Algorithm}.
\newblock Symposium on Simplicity in Algorithms (SOSA), SIAM, 2020, pp.~91--98.

\bibitem{Chase1995}
{\sc Chase, C.~M., and Garg, V.~K.}
\newblock Efficient detection of restricted classes of global predicates.
\newblock In {\em Distributed Algorithms\/} (Berlin, Heidelberg, 1995), J.-M.
  H{\'e}lary and M.~Raynal, Eds., Springer Berlin Heidelberg, pp.~303--317.

\bibitem{Checco2017}
{\sc Checco, A., and Leith, D.~J.}
\newblock Fast, responsive decentralized graph coloring.
\newblock {\em IEEE/ACM Transactions on Networking 25}, 6 (2017), 3628--3640.

\bibitem{Chiu2014}
{\sc Chiu, W.~Y., Chen, C., and Tsai, S.-Y.}
\newblock A 4n-move self-stabilizing algorithm for the minimal dominating set
  problem using an unfair distributed daemon.
\newblock {\em Information Processing Letters 114}, 10 (2014), 515--518.

\bibitem{Duffy2013}
{\sc Duffy, K.~R., Bordenave, C., and Leith, D.~J.}
\newblock Decentralized constraint satisfaction.
\newblock {\em IEEE/ACM Transactions on Networking 21}, 4 (2013), 1298--1308.

\bibitem{Duffy2008}
{\sc Duffy, K.~R., O'Connell, N., and Sapozhnikov, A.}
\newblock Complexity analysis of a decentralised graph colouring algorithm.
\newblock {\em Inf. Process. Lett. 107}, 2 (jul 2008), 60–63.

\bibitem{Fink1985}
{\sc Fink, J.~F., and Jacobson, M.~S.}
\newblock {\em N-Domination in Graphs}.
\newblock John Wiley \&amp; Sons, Inc., USA, 1985, p.~283–300.

\bibitem{Galan2017}
{\sc Gal{\'a}n, S.~F.}
\newblock Simple decentralized graph coloring.
\newblock {\em Computational Optimization and Applications 66}, 1 (Jan 2017),
  163--185.

\bibitem{Garg2022}
{\sc Garg, V.}
\newblock A lattice linear predicate parallel algorithm for the dynamic
  programming problems.
\newblock In {\em 23rd International Conference on Distributed Computing and
  Networking\/} (New York, NY, USA, 2022), ICDCN 2022, Association for
  Computing Machinery, p.~72–76.

\bibitem{Garg2020}
{\sc Garg, V.~K.}
\newblock Predicate detection to solve combinatorial optimization problems.
\newblock In {\em Proceedings of the 32nd ACM Symposium on Parallelism in
  Algorithms and Architectures\/} (New York, NY, USA, 2020), SPAA '20,
  Association for Computing Machinery, p.~235–245.

\bibitem{Garg2021}
{\sc Garg, V.~K.}
\newblock A lattice linear predicate parallel algorithm for the housing market
  problem.
\newblock In {\em Stabilization, Safety, and Security of Distributed Systems\/}
  (Cham, 2021), C.~Johnen, E.~M. Schiller, and S.~Schmid, Eds., Springer
  International Publishing, pp.~108--122.

\bibitem{GODDARD2008}
{\sc Goddard, W., Hedetniemi, S.~T., Jacobs, D.~P., Srimani, P.~K., and Xu, Z.}
\newblock Self-stabilizing graph protocols.
\newblock {\em Parallel Processing Letters 18}, 01 (2008), 189--199.

\bibitem{Guellati2010}
{\sc Guellati, N., and Kheddouci, H.}
\newblock A survey on self-stabilizing algorithms for independence, domination,
  coloring, and matching in graphs.
\newblock {\em J. Parallel Distrib. Comput. 70}, 4 (Apr. 2010), 406–415.

\bibitem{Gupta2021}
{\sc Gupta, A.~T., and Kulkarni, S.~S.}
\newblock Extending lattice linearity for self-stabilizing algorithms.
\newblock In {\em Stabilization, Safety, and Security of Distributed Systems\/}
  (Cham, 2021), C.~Johnen, E.~M. Schiller, and S.~Schmid, Eds., Springer
  International Publishing, pp.~365--379.

\bibitem{Hedetniemi2003}
{\sc Hedetniemi, S., Hedetniemi, S., Jacobs, D., and Srimani, P.}
\newblock Self-stabilizing algorithms for minimal dominating sets and maximal
  independent sets.
\newblock {\em Computers \& Mathematics with Applications 46}, 5 (2003),
  805--811.

\bibitem{Kobayashi2017}
{\sc Kobayashi, H., Kakugawa, H., and Masuzawa, T.}
\newblock Brief announcement: A self-stabilizing algorithm for the minimal
  generalized dominating set problem.
\newblock In {\em Stabilization, Safety, and Security of Distributed Systems\/}
  (Cham, 2017), P.~Spirakis and P.~Tsigas, Eds., Springer International
  Publishing, pp.~378--383.

\bibitem{Kobayashi2022}
{\sc Kobayashi, H., Sudo, Y., Kakugawa, H., and Masuzawa, T.}
\newblock {A Self-Stabilizing Distributed Algorithm for the Generalized
  Dominating Set Problem With Safe Convergence}.
\newblock {\em The Computer Journal\/} (03 2022).
\newblock bxac021.

\bibitem{Leith2006}
{\sc Leith, D.~J., and Clifford, P.}
\newblock Convergence of distributed learning algorithms for optimal wireless
  channel allocation.
\newblock In {\em in Proceedings of IEEE Conference on Decision and Control\/}
  (2006), pp.~2980--2985.

\bibitem{Maruyama2022}
{\sc Maruyama, S., Sudo, Y., Kamei, S., and Kakugawa, H.}
\newblock A self-stabilizing 2-minimal dominating set algorithm based on loop
  composition in networks of girth at least 7.
\newblock In {\em 2022 {IEEE} International Parallel and Distributed Processing
  Symposium ({IPDPS})\/} (May 2022), {IEEE}.

\bibitem{Motskin2009}
{\sc Motskin, A., Roughgarden, T., Skraba, P., and Guibas, L.}
\newblock Lightweight coloring and desynchronization for networks.
\newblock In {\em IEEE INFOCOM 2009\/} (2009), pp.~2383--2391.

\bibitem{Astrand2010}
{\sc \r{A}strand, M., and Suomela, J.}
\newblock Fast distributed approximation algorithms for vertex cover and set
  cover in anonymous networks.
\newblock In {\em Proceedings of the Twenty-Second Annual ACM Symposium on
  Parallelism in Algorithms and Architectures\/} (New York, NY, USA, 2010),
  SPAA '10, Association for Computing Machinery, p.~294–302.

\bibitem{Turau2007}
{\sc Turau, V.}
\newblock Linear self-stabilizing algorithms for the independent and dominating
  set problems using an unfair distributed scheduler.
\newblock {\em Information Processing Letters 103}, 3 (2007), 88--93.

\bibitem{TURAU2010}
{\sc Turau, V.}
\newblock Self-stabilizing vertex cover in anonymous networks with optimal
  approximation ratio.
\newblock {\em Parallel Processing Letters 20}, 02 (2010), 173--186.

\bibitem{Xu2003}
{\sc Xu, Z., Hedetniemi, S.~T., Goddard, W., and Srimani, P.~K.}
\newblock A synchronous self-stabilizing minimal domination protocol in an
  arbitrary network graph.
\newblock In {\em Distributed Computing - IWDC 2003\/} (Berlin, Heidelberg,
  2003), S.~R. Das and S.~K. Das, Eds., Springer Berlin Heidelberg, pp.~26--32.

\end{thebibliography}
\bibliographystyle{acm}

\end{document}